\documentclass[a4paper,onecolumn,accepted=2024-02-25]{quantumarticle}
\pdfoutput=1
\usepackage[utf8]{inputenc}
\usepackage[T1]{fontenc}
\usepackage[english]{babel}
\usepackage{amsmath,amssymb,amsthm}
\usepackage{hyperref}
\usepackage{url}
\usepackage{subcaption}
\usepackage{mathtools}
\usepackage{booktabs}
\usepackage{csquotes}
\usepackage[block=space,
backend=biber,
bibstyle=alphabetic,
citestyle=alphabetic,
maxbibnames=10,
]{biblatex}
\renewbibmacro{in:}{}
\bibsetup{
  \setcounter{abbrvpenalty}{9000}
  \setcounter{highnamepenalty}{9000}
  \setcounter{lownamepenalty}{9000}
  \setcounter{biburlnumpenalty}{9000}
}

\newtheorem{question}{Question}

\newtheorem{prop}{Proposition}

\newcommand{\chsh}{\mathcal{C}}

\newcommand{\SP}{\emph{SP }}
\newcommand{\LP}{{\emph{LP }}}
\newcommand{\zs}{{\mathtt{S}}}
\newcommand{\zl}{{\mathtt{L}}}
\newcommand{\Q}{{\mathcal{Q}}}
\newcommand{\SRQ}{{\mathcal{Q}_{SR}}}
\newcommand{\LRQ}{{\mathcal{Q}_{LR}}}
\newcommand{\QC}{{\mathcal{M}_{qc}}}
\newcommand{\QQ}{{\mathcal{M}_{qq}}}

\DeclareMathSymbol{\shortminus}{\mathbin}{AMSa}{"39}

\usepackage{tikz}
\usetikzlibrary{automata,arrows.meta,positioning,quotes}
\usetikzlibrary{calc}
\usetikzlibrary{shapes.symbols}
\usetikzlibrary{fit}
\usetikzlibrary{shapes.geometric}
\usepackage{siunitx}
\usepackage{physics}

\usepackage{subcaption} 

\AtEveryBibitem{\clearfield{month}}
\AtEveryBibitem{\clearfield{issue}}
\AtEveryBibitem{\clearfield{primaryclass}}
\AtEveryBibitem{\clearfield{archivePrefix}}

\begin{document}
\definecolor{black}{rgb}{0,0,0}

\title{Certifying long-range quantum correlations through routed Bell tests}
\author{Edwin Peter Lobo}
\author{Jef Pauwels}
\author{Stefano Pironio}
\affiliation{\normalsize Laboratoire d'Information Quantique, Université libre de Bruxelles (ULB), Belgium}
\date{24 April 2024}

\maketitle
\begin{abstract}
    Losses in the transmission channel, which increase with distance, pose a major obstacle to photonics demonstrations of quantum nonlocality and its applications.
    Recently, Chaturvedi, Viola, and Pawlowski (CVP) [arXiv:2211.14231] introduced a variation of standard Bell experiments with the goal of extending the range over which quantum nonlocality can be demonstrated.
    These experiments, which we call `routed Bell experiments', involve two distant parties, Alice and Bob, and allow Bob to route his quantum particle along two possible paths and measure it at two distinct locations – one near and another far from the source.
    The premise is that a high-quality Bell violation in the short-path should constrain the possible strategies underlying the experiment, thereby weakening the conditions required to detect nonlocal correlations in the long-path. Building on this idea, CVP showed that there are certain quantum correlations in routed Bell experiments such that the outcomes of the remote measurement device cannot be classically predetermined, even when its detection efficiency is arbitrarily low. 
    In this paper, we show that the correlations considered by CVP, though they cannot be classically predetermined, do not require the transmission of quantum systems to the remote measurement device. 
    This leads us to define and formalize the concept of `short-range' and `long-range' quantum correlations in routed Bell experiments. We show that these correlations can be characterized through standard semidefinite-programming hierarchies for non-commutative polynomial optimization. We then explore the conditions under which short-range quantum correlations can be ruled out and long-range quantum nonlocality can be certified in routed Bell experiments. We point out that there exist fundamental lower-bounds on the critical detection efficiency of the distant measurement device, implying that routed Bell experiments cannot demonstrate long-range quantum nonlocality at arbitrarily large distances. However, we do find that routed Bell experiments allow for reducing the detection efficiency threshold necessary to certify long-range quantum correlations. The improvements, though, are significantly smaller than those suggested by CVP's analysis. 
\end{abstract}

\section{Introduction}

Losses in the transmission channel, which increase with the distance 
(exponentially in fibers and quadratically in free-space),  are a major obstacle for demonstrating the 
violation of Bell inequalities in photonic experiments 
\cite{Pearle1970,Clauser1974,Garg1987,Eberhard1993,Larsson1998,Massar2003,Brunner2014}. They represent a daunting challenge for long-distance applications that rely on quantum nonlocality, such as entanglement certification between distant parties or device-independent quantum key distribution (DIQKD) \cite{Acin2007,zapatero2023advances}.

Chaturvedi, Viola, and Pawlowski (CVP) recently proposed an interesting idea that could potentially address this limitation \cite{Chaturvedi2022}. They looked at a Bell experiment that involves the usual two distant parties, Alice and Bob. However, in their setup, Bob can measure his quantum particles at two distinct locations – one close to the source, $B_\zs$, and another far away, $B_\zl$, as illustrated in Fig.~1. ($\zs$ and $\zl$ stand for `short' and `long' distance, respectively). This can be accomplished, for example, by using a switch that directs Bob's quantum particle either to the nearby measurement device $B_\zs$ or to the distant one $B_\zl$, depending on a classical input $z\in\{\zs,\zl\}$. Like all other components of the experiment (the source, the transmission channel, the measurement devices), the switch does not need to be trusted. As in any Bell experiment, certain causality constraints are assumed: events at Alice's side cannot causally influence those at Bob's side, and vice versa. In particular, the switch input $z$ cannot affect the quantum state or measurement of Alice. 
We refer to such Bell experiments with selective routing of quantum particles to different locations as \emph{routed Bell experiments}.

\begin{figure}[t]
    \centering
    \begin{subfigure}{0.485\textwidth}
    \includegraphics[width=\textwidth]{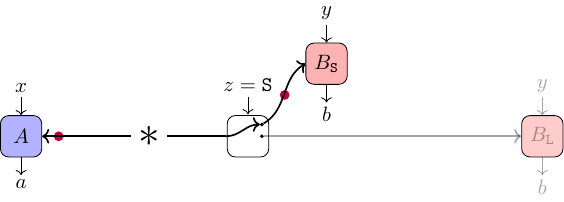}
     \caption{The particle is routed to the closeby device when $z=\zs$.}
    \end{subfigure}
    \hfill
    \begin{subfigure}{0.485\textwidth}
    \includegraphics[width=\textwidth]{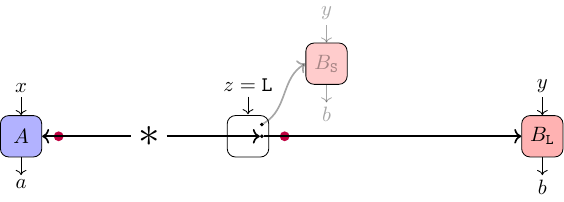}
    \caption{The particle is routed to the distant device when $z=\zl$.}
    \end{subfigure}
    \caption{Routed Bell experiment. Depending on the value $z\in\{\zs,\zl\}$, a switch directs Bob's quantum particle either to (a) a nearby measurement device $B_\zs$ or (b) to a distant one $B_\zl$. The experiment is characterized by the joint output|input probabilities $p(a,b|x,y,z)$. \label{fig:routedBell} }
\end{figure}

In their paper, CVP examine a situation where Alice's measurement device $A$ receives a binary input $x\in\{0,1\}$ and produces a binary outcome $a\in\{\pm 1\}$. Similarly, the measurement device $B_\zs$ or $B_\zl$, selected by the switch, receives a binary input $y\in\{0,1\}$ and produces a binary output $b\in\{\pm 1\}$. 
We denote Alice's observables as $A_x$, those of Bob as  $B_{yz}$ (which may depend on the switch setting $z$), and the expectation of their product as $\langle A_x B_{yz}\rangle$. Alice and Bob can then evaluate two CHSH expressions,
\begin{equation}
\label{eq:CHSH}
\chsh_z = \langle A_0 B_{0z}\rangle + \langle A_0 B_{1z}\rangle + \langle A_1 B_{0z}\rangle - \langle A_1 B_{1z}\rangle\,,
\end{equation}
involving either the nearby or faraway measurement devices, depending on the switch value $z\in\{\zs,\zl\}$. We refer to $\chsh_\zs$ as the short-path (\emph{SP}) CHSH value and $\chsh_\zl$ as the long-path (\emph{LP}) CHSH value.

In a standard Bell experiment, the condition for ruling out classical models and thus certifying genuine quantum properties in 
the \LP test would be the violation of the well-known local bound: $\chsh_\zl\leq  2$. 
CVP argue, however, that in any quantum model where the output of the distant measurement device $B_\zl$ is predetermined 
by classical variables, the \LP value is upper-bounded by
\begin{equation}
    \label{eq:CHSHtradeoff}
    \chsh_\zl \leq \sqrt{8 - \chsh_\zs^{\;2}}\,,
\end{equation}
whenever the \SP value violates the local bound, i.e., whenever $\chsh_\zs>2$. This bound implies that the \SP test can be used to weaken the conditions for ruling out classical models in the \LP test, since the right-hand side of (\ref{eq:CHSHtradeoff}) is strictly smaller than 2 for any value of $\chsh_\zs>2$. 

This result suggests that routed Bell experiments might provide a way to dramatically extend  
the range over which nonlocality can be demonstrated. Indeed, assume, as an illustration, an 
ideal CHSH realization where the source prepares the maximally entangled two-qubit state $|\phi_+\rangle$, $A$ measures in the bases $Z,X$, and $B_\zs$ and $B_\zl$ in the bases $(Z\pm X)/\sqrt{2}$. Then $\chsh_\zs = \chsh_\zl = 2\sqrt{2}$. 
Typically, however, the \LP CHSH value will be lower than the \SP CHSH value due to additional 
losses and noise in the transmission channel. For instance, let's assume that $B_\zs$ has a device 
with global detection efficiency $\eta_\zs$, while $B_\zl$, being farther away from the source, 
has a smaller efficiency $\eta_\zl<\eta_\zs$ (and for simplicity that $A$ has a measurement device with 
unit detection efficiency $\eta_A=1$). Hence, with some non-zero probability, $B_\zs$ and $B_\zl$
 will occasionally fail to click. Simply discarding the `no-click' outcomes $\varnothing$ can lead 
 to the detection loophole and is only valid under the fair sampling assumption \cite{Pearle1970,Clauser1974}. 
 However, one can deal with these `no-click' results by mapping them to one of the  $\pm 1$ 
 outcomes, say, $+1$ \cite{Clauser1974,Wilms2008,Brunner2014}. Taking into account losses in this way, the \SP and \LP CHSH values become
\begin{equation}
    \label{eq:CHSHefficiencies}
    \chsh_\zs  = \eta_\zs 2\sqrt{2} \quad \text{ and } \quad \chsh_\zl  = \eta_\zl 2\sqrt{2}\,.
\end{equation}
If we substitute these values into (\ref{eq:CHSHtradeoff}), we find that classical models for $B_\zl$ are ruled out if 
\begin{equation}
    \label{eq:etatradeoff}
    \eta_\zl > \sqrt{1-\eta_\zs^2}\,.
\end{equation}
For instance, if $\eta_\zs = 1-\delta$, then an efficiency $\eta_\zl > \sqrt{2\delta}$ for the far-away device is sufficient and can be made arbitrarily small as $\delta\rightarrow 0$. Taking into account that detection efficiencies decrease with transmission distance, the implication is that by performing high-quality CHSH tests close to the source (which are achievable with current technology), the bound \eqref{eq:CHSHtradeoff} can be violated even if the measurement device $B_\zl$ is at an arbitrarily large distance from the source.

The significance of this result, however, hinges on the assumptions used to derive the bound \eqref{eq:CHSHtradeoff} and in particular on what one means by `demonstrating nonlocality' and `ruling out classical models' in routed Bell experiments.  Evidently, the aim is not to rule out local hidden-variable models à la Bell for the \emph{entire} routed Bell experiment. This is because such models are already ruled out by the \SP test, without any need to even consider the \LP test. Furthermore, the relation \eqref{eq:CHSHtradeoff} explicitly assumes that there is no local hidden-variable model for the \SP test, since it relies on the violation $\chsh_\zs>2$ of the local bound.

The idea is thus to take for granted that the devices $A$ and $B_\zs$, which are located close to the source, behave quantumly and ask whether the observed correlations can, or cannot, be explained if the faraway device $B_\zl$ behaves classically. However, various definitions are possible for what it means for $B_\zl$ to behave classically.

In this paper, we adopt the following view. We assume that on Bob's side the transmission of quantum information is only possible at short distance, where by `short distance' we mean that it cannot reach the remote device $B_\zl$. Thus, beyond some point on the line connecting the source to $B_\zl$, quantum data can no longer be transmitted or processed and the experimental setup becomes entirely classical. In particular, the device $B_\zl$ functions as a purely classical device (e.g., a classical computer) receiving classical data through a classical transmission channel (e.g., radio waves from a wifi emitter). If such a model can reproduce the experimental data, then it is not possible to claim in a device-independent way that long-range quantum correlations have been demonstrated, since they can be replicated without any entanglement or quantum communication reaching $B_\zl$. Conversely, if no such model can reproduce the experimental data, then long-range quantum correlations can be certified. 

We view this framing of the problem as the relevant one in the context of transmission losses and their impact on the demonstration of quantum nonlocality over long distances.
Indeed, if at short distance quantum nonlocality is already established and quantum resources are known to be transmitted  (to perform the \SP CHSH test), then the pertinent question is whether transmission of quantum resources is also necessary over longer distances, or whether it is possible to achieve the same results without them.

The above standpoint differs from the one adopted by CVP. CVP's definition of what it means for the distant device $B_\zl$ to behave classically, and which is used to derive the bound \eqref{eq:CHSHtradeoff}, is that $B_\zl$'s outcome is determined by classical variables already set at the source. We point out in this paper that this definition does not encompass the most general class of short-range quantum models in the sense outlined above. We do this by presenting a simple strategy in which the remote device $B_\zl$ is entirely classical and no quantum information ever reaches it, yet which achieves the standard classical bound $\chsh_\zl=2$ for any value of $\chsh_\zs>2$, i.e., that violates the bound \eqref{eq:CHSHtradeoff} satisfied by CVP models.

This motivates us to introduce a proper definition of \emph{short-range quantum correlations}, which suitably captures long-range nonlocality and the potential tradeoff between \SP and \LP tests in routed Bell experiments. We then show that short-range quantum correlations (and incidentally the correlations considered by CVP) can be characterized through standard semidefinite-programming hierarchies for non-commutative polynomial optimization. 
Based on this formulation, we derive new \LP bounds and show that \SP tests do lead to weakened conditions for \LP tests.  Although they allow for reducing the detection efficiency threshold necessary to certify long-range quantum correlations, the improvements are considerably smaller than those suggested by CVP's analysis and the bound \eqref{eq:CHSHtradeoff}. In particular, we point out that there exists fundamental lower-bounds on the critical detection efficiency $\eta_\zl$ of the distant measurement device $B_\zl$, namely $\eta_\zl>1/M$ where $M$ is the number of measurement settings of $B_\zl$. The same lower-bounds hold for regular Bell tests and imply that routed Bell experiments cannot be used to demonstrate long-range quantum nonlocality at arbitrarily large distances. 

This paper is organized as follows. We first briefly review, in Section~\ref{sec:CVPmodels}, the class of models considered by CVP and present an example of a short-range quantum correlation that violates the bound \eqref{eq:CHSHtradeoff}. The reader not interested in an extensive analysis of CVP models may skip Section~\ref{sec:CVPmodels} without loss of continuity. In Section~\ref{sec:RoutedBellSRQ}, we introduce more formally our definition of short-range and long-range quantum correlations in routed Bell experiments. We also analyze more generally various classes of correlations that can be obtained in routed Bell experiments and show how they can be characterized through standard semidefinite-programming hierarchies for non-commutative polynomial optimization. In Section~\ref{sec:SPvsLP}, we derive new \emph{SP}/\LP relations valid according to our definition. In Section~\ref{sec:dec}, we analyze the required detection efficiencies required to demonstrate long-range quantum correlations in routed Bell experiments. We conclude with a discussion of our results.

\section{CVP models vs short-range quantum correlations} \label{sec:CVPmodels}

As they form the initial motivation for the present paper, we begin by examining CVP models as a potential mechanism for the observed correlations in routed Bell experiments.

Assume, for concreteness, a routed Bell experiment characterized by quantum correlations as in  \eqref{eq:CHSHefficiencies} where $\eta_\zs$ and $\eta_\zl$ are such that $\chsh_\zs>2$, but $\chsh_\zl \le 2$. A natural question to ask is: Can such correlations be reproduced by a classical remote device $B_\zl$ without any distribution of entanglement between $A$ and $B_\zl$? In a standard Bell experiment, this would be the case if and only if the output of $B_\zl$ were fully determined by classical variables $\lambda$ already set at the source and shared with $A$. Hence, it seems reasonable to make the same assumption here.
However, since the \SP test does violate the CHSH inequality, the measurement devices $A$ and $B_\zs$ must, as discussed earlier, share and exploit quantum entanglement.

This leads us to consider a hybrid quantum-classical model where the source generates an entangled quantum state $\rho_{AB_\zs}$ that can reach $A$ and $B_\zs$, along with classical variables $\lambda$ that determine $B_\zl$'s measurement outcomes. In full generality, these classical variables can also be correlated with the quantum system of $A$ and $B_\zs$ and (partly) determine their outcomes, i.e., the state $\rho^\lambda_{AB_\zs}$ can also depend on $\lambda$. For such a model, we can then write the correlations generated in a routed Bell experiment as
\begin{equation}
    \label{eq:CVPcorrelations1}
p(a,b|x,y,z) =\\
    \begin{cases}
         \sum_\lambda p(\lambda)\,\Tr\left[\rho^\lambda_{AB_\zs}\, M_{a|x}\otimes M_{b|y,\zs}\right] & \text{ if } z=\zs\\
        \sum_\lambda p(\lambda)\, p(b|y,\lambda)\, p(a|x,\lambda)& \text{ if } z=\zl
    \end{cases}
\end{equation}
where the first line describes quantum correlations between $A$ and $B_\zs$ and the second line classical correlations between $A$ and $B_\zl$. 
These two lines should be coupled by the condition that $p(a|x,\lambda) = \Tr  \left[\rho^\lambda_{A}\, M_{a|x}\right]$, since 
what Alice does cannot causally depend on what happens on Bob's side, and in particular on whether $z=\zs$ or $\zl$. This is the formulation used in \cite{Chaturvedi2022}, which leads to the bound \eqref{eq:CHSHtradeoff}.

The intuition behind the derivation of this bound is as follows.
Assume first for simplicity that the \SP CHSH expression reaches the maximal quantum value $\chsh_\zs=2\sqrt{2}$. Then, by standard self-testing 
results \cite{Cirelson1980, Mayers2004, Supic2020}, it can be inferred that the 
measurement $A_x$ corresponds to a Pauli measurement on a two-dimensional subspace of $A$ that is maximally entangled with $B_\zs$ and acts as the 
identity on any other degrees of freedom. In particular, the measurement outcome of $A_x$ must be fully random and uncorrelated with the classical 
instructions $\lambda$ shared with $B_\zl$. Consequently, we have $p(a|x,\lambda)=p(a|x)=1/2$ for all $\lambda$. Substituting this 
condition in \eqref{eq:CVPcorrelations1} implies that the correlations between $a$ and the output $b$ of $B_\zl$ 
vanishes: $\langle A_x B_{y\zl}\rangle=\sum_{a,b\in\{\pm 1\}} ab\, p(a,b|x,y,\zl)=\frac{1}{2}\sum_{a,b\in\{\pm 1\}} ab \,p(b|y,\zl)=0$ for all $x,y$. This in turn implies that $\chsh_\zl=0$.

The hypothesis $\chsh_\zs=2\sqrt{2}$ is obviously too strong in any real-life experiment. However, the above argument can be refined using the fact that for any value $\chsh_\zs>2$, there is a bound on how much the measurement outcomes of $A_x$ can be correlated to any other system besides $B_\zs$, and in particular to the classical instructions $\lambda$ shared with $B_\zl$: specifically $\left|p(a|x,\lambda)\right|\leq 1/2 +\sqrt{8-\chsh_\zs^2}/4$ for all $\lambda$ \cite{Pironio2010}. Building on this result, CVP arrive at the bound \eqref{eq:CHSHtradeoff}.

\subsection{A strategy based on a fully classical $B_\zl$ that violates the bound \eqref{eq:CHSHtradeoff}}\label{sec:example}

We now present a simple strategy where $B_\zl$ is entirely classical and which aligns with the intuitive notion of short-range quantum correlations discussed in the introduction, but that violates \eqref{eq:CHSHtradeoff}.  This shows that CVP models do not correspond to the notion taken here of what it means for $B_\zl$ to behave classically.

We start from the ordinary quantum strategy that yields the \SP and \LP CHSH expectations (\ref{eq:CHSHefficiencies}). We recall that in this strategy, the source prepares the two-qubit entangled state $(|00\rangle+|11\rangle)/\sqrt{2}$. The first qubit is measured by $A$ in the bases ${Z,X}$, while the second qubit is directed to either $B_\zs$ or $B_\zl$, depending on the switch setting, and is then measured in the bases $(Z\pm X)/\sqrt{2}$.  If the second qubit is directed to $B_\zs$, it has a probability $\eta_\zs$ of being detected, whereas if it is directed to $B_\zl$, it has a lower probability $\eta_\zl$ of being detected.


Consider now an alternative strategy where the source, the measurement device $A$, the switch, and the measurement device $B_\zs$  
all behave as in the ordinary quantum strategy described above. Thus any value $\chsh_\zs\in[0,2\sqrt{2}]$ can be obtained by 
tuning $\eta_\zs$. We only modify what happens in the experiment \emph{after} the second qubit has been directed towards $B_\zl$ 
when the switch has been set to $z=\zl$. In this case, at some location between the switch and $B_\zl$ -- possibly just after 
the switch, but in any case before reaching $B_\zl$ -- the second qubit gets measured in the $Z$ basis yielding an 
outcome $\lambda\in\{\pm 1\}$, as illustrated in Fig. \ref{fig:CVP_counterexample}. This classical outcome is then 
transmitted to $B_\zl$ through some purely classical channel and upon receiving $\lambda$, $B_\zl$ simply outputs it, 
irrespective of which input $y\in\{0,1\}$ is selected. We then have $p(\lambda)=1/2$, $p(b|y,\lambda) = \delta_{b,\lambda}$, 
$p(a|x,\lambda) = \Tr\left[\rho_\lambda A_x\right]$, where $\rho_{1} = |0\rangle\langle 0|$ and $\rho_{-1} = |1\rangle\langle 1|$ 
are the reduced states for Alice conditioned upon $\lambda$. We can then evaluate the probabilities $p(a,b,|x,y,\zl)$ by inserting these 
expressions in \eqref{eq:CVPcorrelations1} or more simply directly evaluate the \LP correlators $\langle A_xB_{y\zl}\rangle$ through
\begin{equation}
    \langle A_xB_{y\zl}\rangle = \langle\phi_+| A_x \otimes Z |\phi_+\rangle\,,
\end{equation}
since whatever the choice of $y$, the `measurement' $B_{y\zl}$ corresponds to an effective $Z$ measurement on Bob's particle.
If $x=0$, i.e, $A_0=Z$, we find $\langle A_0B_{y\zl}\rangle =1$, while if $x=1$, i.e, $A_1=X$, we find $\langle A_1B_{y\zl}\rangle =0$. As claimed, this implies the \LP CHSH value $\chsh_\zl =2$.

\begin{figure}[t]
    \centering
   \includegraphics{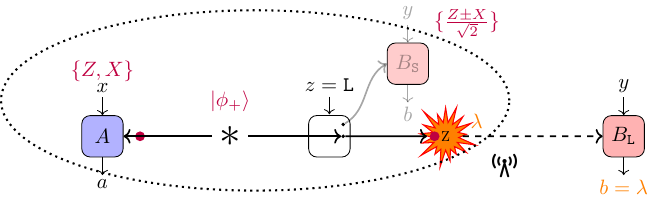}
    \caption{A strategy yielding both $\chsh_\zs=2\sqrt{2}$ and $\chsh_\zl=2$. The source prepares the Bell state $|\phi_+\rangle$, Alice performs $Z$ or $X$ measurements on her qubit, and if $z=\zs$, Bob performs the measurements $(Z\pm X)/\sqrt{2}$ at $B_\zs$. On the other hand, if $z=\zl$, Bob's qubit gets measured in the $Z$ basis and the classical outcome $\lambda$ is transmitted to $B_\zl$, which simply outputs it. The entire quantum part of this experiment, enclosed by the dotted line, may happen, for instance, on an optical table in Gdansk and the classical outcome $\lambda$ sent by email to $B_\zl$ located in Sydney. }
    \label{fig:CVP_counterexample}
\end{figure}

\subsection{Discussion}\label{CVPresources}
The above example shows that the notion of short-range quantum  strategies considered in the present paper does not coincide with the set of CVP strategies. The assumption that the outcomes of $B_\zl$ are determined by classical variables that are already specified at the source is too strong to capture properly the notion of short-range quantum correlations and to detect long-range quantum nonlocality.

If one were to take CVP's original definition as a definition for long-range nonlocality in routed Bell experiments, then one would have to agree that a standard Bell experiment done on an optical table in Gdansk and whose classical measurement outcomes are sent by email to Sydney would feature nonlocality between Gdansk and Sydney, since such a procedure would allow for implementing the strategy of Fig.~\ref{fig:CVP_counterexample} which violates the CVP bound \eqref{eq:CHSHtradeoff}.

In particular, the above example shows that from the observation of a violation of the bound \eqref{eq:CHSHtradeoff}, it \emph{cannot} be inferred that  
\begin{itemize}
    \setlength{\itemsep}{0pt}
    \setlength{\parskip}{0pt}
    \setlength{\parsep}{0pt}
    \item long-range quantum resources are necessary to reproduce the correlations,
    \item entanglement had to be distributed from the source to the faraway device $B_\zl$,
    \item the device $B_\zl$ behaves quantumly, e.g., it performs incompatible measurements, 
    \item fresh quantum randomness is generated after the input $y$ is given to $B_\zl$,
    \item the correlations can be used for secure DIQKD between $A$ and $B_\zl$,
    \item or that any other property typically associated with quantum nonlocality is present at the remote device $B_\zl$.
\end{itemize}

There are actually two distinct questions that can be raised regarding the correlations observed in a routed Bell experiment.
\begin{question}
    Can we trace back the outcomes of the remote device $B_\zl$ to a genuine quantum measurement? 
\end{question}
\begin{question}
     Can we trace back the outcomes of the remote device $B_\zl$ to a genuine quantum measurement occurring at or near $B_\zl$?
\end{question}
Question 1 is addressed by CVP models. If it is not possible to account for the outcomes of the remote device by classical variables predetermined at the source, then some quantum measurement must have taken place between the source and $B_\zl$. However, CVP models say nothing about \emph{where} this quantum measurement happened. It could have taken place near the source, the switch, or in proximity to $B_\zl$.

Question 2, on the other hand, focuses not only on the fact that a quantum measurement took place, but in addition that it took place at the remote location $B_\zl$.  We view this question as the relevant one in the context of routed Bell experiments and the impact of transmission losses on the demonstration of quantum nonlocality over long distances.
Indeed, in a routed Bell experiment that exhibits a \SP violation, it is already clear that the experiment possesses the capability to demonstrate quantum effects at short distances. The interesting question in most applications is not just whether the outcomes of $B_\zl$ also depend on such quantum effects, but whether they depend on quantum effects arising far away from the source.

In the Gdansk-Sydney experiment, it is true that based on the information available in Sydney,  and if the bound \eqref{eq:CHSHtradeoff} is 
violated, one can infer that a quantum measurement must have taken place on the particle sent by the source. But this does not imply that 
nonlocality has been established between Gdansk and Sydney. In particular, the measurement performed on the particle sent by the source might 
have happened well before the input $y$ was provided to $B_\zl$ in Sydney. In contrast, the objective of the present paper is to identify 
conditions under which one can conclude that a quantum measurement took place in Sydney after the input $y$ was provided to $B_\zl$. The  
distinction between these scenarios is illustrated in Fig.~\ref{fig:spacetime_diagrams}, which depicts a spacetime diagram representing three types of correlations 
that can be observed in a routed Bell experiment: genuine long-range quantum correlations, short-range quantum correlations, and CVP correlations.

\begin{figure}[t] 
    \begin{subfigure}{.332\textwidth}
        \includegraphics[width=1\textwidth]{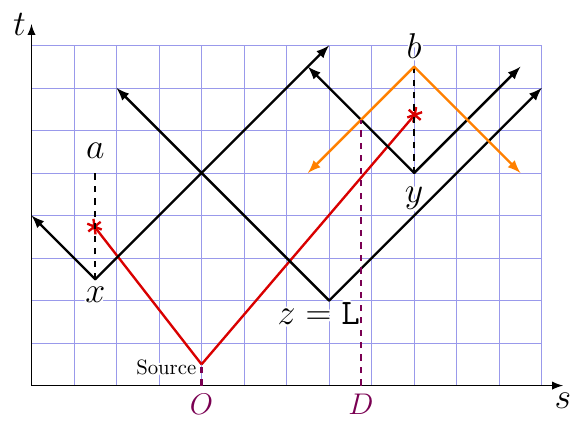} 
        \caption{Long-Range Quantum Correlations}
    \end{subfigure}%
    \begin{subfigure}{.332\textwidth}
        \includegraphics[width=1\textwidth]{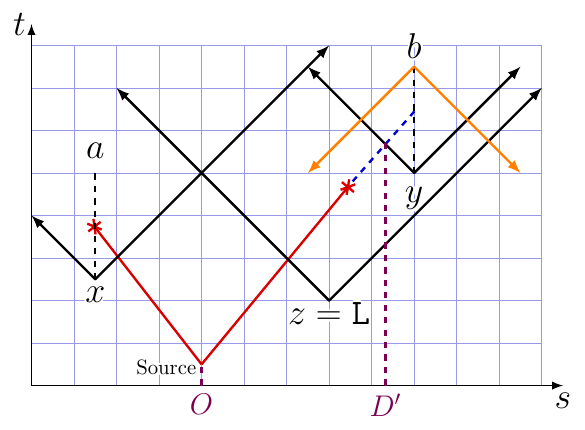} 
        \caption{Short-Range Quantum Correlations}
    \end{subfigure}%
    \begin{subfigure}{.332\textwidth}
        \includegraphics[width=1\textwidth]{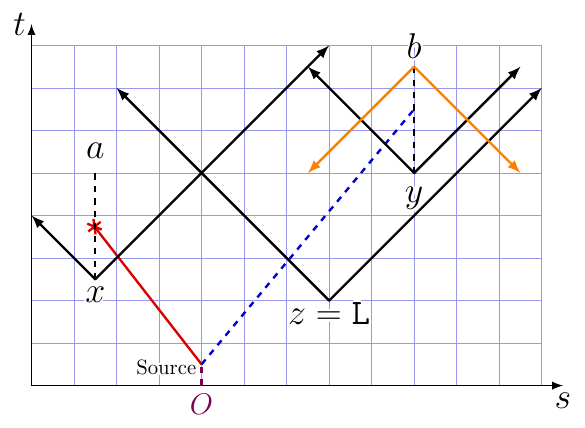} 
        \caption{CVP Correlations}
    \end{subfigure}%
    \caption{Spacetime diagram of routed Bell experiments in the case where $z=\zl$. Red lines represent the transmission of quantum information, 
    blue dotted lines of classical variables.  \emph{(a)} Long-range quantum correlations: a quantum measurement is performed by Bob within the 
    future-light cone of the input choice $y$ and beyond a distance $D$ from the source,  
    corresponding to the point where the future-light cone of $y$ and the past-light cone of $b$ intersect. 
    \emph{(b)} Short-range quantum correlations: a measurement occurs before the particle reaches $B_\zl$, at a distance smaller than $D'$, and outside the future-light cone of $y$.  \emph{c)} CVP correlations: the outcomes of $B_\zl$ are determined by classical variables predetermined at the source. Correlations that cannot be represented by \emph{(c)} models belong to either \emph{(b)} or \emph{(a)} types. Correlations that cannot be represented by \emph{(b)} models are of the \emph{(a)} type.} 
    \label{fig:spacetime_diagrams}
\end{figure}

We note that the strategy depicted in Fig.~\ref{fig:CVP_counterexample} is fully consistent with a device-independent setting, where all components including the switch, the communication channel, and the device $B_\zl$ are untrusted. The intermediate $Z$ measurement that is performed when $z=\zl$, could for instance be implemented by the switch device itself. However, it is worth emphasizing that even if the switch device were to be fully trusted, the strategy could still be applied by performing the intermediate $Z$ measurement somewhere on the transmission channel connecting the switch to $B_\zl$.
Again, this would be fully consistent with device-independent scenarios, where transmission channels between devices are usually assumed to be untrusted and may not behave as expected (in our case, the lossy channel characterizing the transmission line would be replaced by a quantum-classical channel).

One might argue that in a semi-device-independent setting where both the switch and the communication channel are trusted, and losses are nonmalicious, our strategy would not be applicable. However, losses in Bell experiments are a problem only in a scenario where they are untrusted. The detection loophole in standard Bell experiments arises from the possibility for classical models to replace the existing transmission channel with an alternative one where losses are not merely passive events but explicitly depend on the inputs of the devices.
If losses are instead assumed to be innocent and cannot be exploited by underlying classical models, or can only be exploited in a limited way \cite{Putz2016} -- this is the famous fair sampling assumption \cite{Pearle1970,Clauser1974} that one typically aims to avoid when analyzing Bell experiments -- then nonlocality can already be demonstrated in standard Bell experiments regardless of the extent of such losses.  Thus if the aim is to establish nonlocality over long distances, there is no clear incentive for considering routed Bell experiments with trusted transmission channels. 

But even in a hypothetical scenario where both the switch and the transmission channel are trusted (or a scenario based on certain assumptions preventing any quantum measurement between the source and the measurement device $B_\zl$), the conclusion reached by the violation of CVP models would still be limited. Indeed, the strategy depicted in Fig.~\ref{fig:CVP_counterexample} could still be applied with the measurement $Z$ taking place inside the measurement device $B_\zl$ independent of the input $y$. Thus, a violation of the bound \eqref{eq:CHSHtradeoff} would indicate that quantum entanglement has been established between $A$ and $B_\zl$, however, it would not imply other quantum properties typically associated with quantum nonlocality, such as the fact that different input choices $y$ correspond to incompatible measurements. 

The strategy of Fig.~\ref{fig:CVP_counterexample} also has implications for an alternate interpretation of routed Bell experiments, discussed in \cite{Chaturvedi2022}. Instead of considering experiments with a switch that alters the path of the quantum particle and routes it to different measurement devices, one could also consider a (more impractical) scenario in which Bob has a single measurement device that is physically moved in each experimental run either to the close location $B_\zs$ or the faraway location $B_\zl$. However, in light of Fig.~\ref{fig:CVP_counterexample}, to achieve a violation of the bound \eqref{eq:CHSHtradeoff}, it is unnecessary to actually move the measurement device. The same violation can be obtained by consistently leaving the measurement device at the close location $B_\zs$, performing the intermediate $Z$ measurement there, and relaying its classical outcome to the remote location $B_\zl$.

Finally, we point out that the strategy depicted in Fig.~\ref{fig:CVP_counterexample} can also be interpreted from a more foundational perspective. One might for instance consider alternative theories to the standard quantum theory, where entangled particles undergo spontaneous collapse after travelling a certain distance. For example, a pair in the $|\phi_+\rangle=\left(|00\rangle+|11\rangle\right)/\sqrt{2}$ state might collapse with probability 1/2 either to the $|00\rangle$ state or the $|11\rangle$ after particles have travelled a distance $d>D$. Such an alternative theory could then violate Bell inequalities when the measurement devices of Alice and Bob are within distance $D$, but would satisfy standard Bell inequalities when they are at a distance larger than $D$. One could attempt to rule out such theories by performing routed Bell experiments with the remote device $B_\zl$ located at a distance $d>D$. Since the spontaneous collapse described above effectively amounts to doing the intermediate $Z$ measurement in Fig.~\ref{fig:CVP_counterexample}, such an alternative theory could in principle reproduce a \LP CHSH value of $\chsh_\zl=2$  and thus cannot be falsified by a violation of the bound \eqref{eq:CHSHtradeoff}.

Because of all the reasons above, we consider in the present paper models that are alternative to those introduced by CVP and which accommodate strategies such as those depicted in Fig.~\ref{fig:CVP_counterexample}.

\subsection{Monogamy of quantum correlations and the bound \eqref{eq:CHSHtradeoff}}\label{sec:monogamy}
Before introducing more formally our formulation of short-range and long-range quantum correlations in routed Bell experiments, we point out that the bound \eqref{eq:CHSHtradeoff} was already established in a more general setting in \cite{Toner2006}. 

This follows from the fact that CVP's assumption that the measurement results of $B_\zl$ are predetermined at the source is equivalent to the assumption that the source prepares a tripartite $qqc$-state
\begin{equation}
    \label{eq:qqc}
    \rho_{AB_\zs B_\zl} = \sum_{\lambda} p(\lambda)\, \rho_{AB_\zs}^\lambda \otimes |\lambda\rangle\langle \lambda|_{B_\zl}
\end{equation}
yielding,  when measurements are performed on $A$ and on either $B_\zs$ or $B_\zl$, correlations of the form
\begin{equation}
    \label{eq:CVPcorrelations}
    p(a,b|x,y,z) =
        \begin{cases}
             \Tr\left[\rho_{AB_\zs B_\zl}\, M_{a|x}\otimes M_{b|y,\zs}\otimes \mathbb{I}\right] & \text{ if } z=\zs\\
            \Tr \left[\rho_{AB_\zs B_\zl}\, M_{a|x}\otimes \mathbb{I}\otimes M_{b|y,\zl}\right] & \text{ if } z=\zl\,.
        \end{cases}
\end{equation}
Indeed, using the explicit $qqc$-form \eqref{eq:qqc} of the state $\rho_{AB_\zs B_\zl}$, it is easily seen that the above correlations are equivalent to those in~\eqref{eq:CVPcorrelations1}.

More generally, one could consider correlations obtained by measuring a genuine tripartite $qqq$-state $\rho_{AB_\zs B_\zl}$. Such correlations would correspond to a routed Bell experiment where the source produces on Bob's side a pair of quantum systems $(B_\zs,B_\zl)$, but where the nearby measurement device only measures the $B_\zs$ system and the remote measurement device only measures the $B_\zl$ system. 
It was already proven in \cite{Toner2006} that the relation \eqref{eq:CHSHtradeoff} holds for correlations arising from measurements on such  tripartite $qqq$-state $\rho_{AB_\zs B_\zl}$, and thus also for the more restricted case of $qqc$-states corresponding to CVP correlations. It was furthermore already proven in \cite{Toner2006} (below the proof of Lemma~4) that when $\chsh_\zs>2$, it is sufficient to consider $qqc$-states to saturate the bound \eqref{eq:CHSHtradeoff}. 

In \cite{Toner2006}, the bound \eqref{eq:CHSHtradeoff} is interpreted as a monogamy relation: two pair of systems ($A$ and $B_\zs$) in a general tripartite quantum state can lead to a violation of the CHSH inequality, only if the CHSH value for the other pairs of systems ($A$ and $B_\zl$) is limited, even if $B_\zl$ is quantum. It is indeed easy to see that the intuition behind the derivation of the bound \eqref{eq:CHSHtradeoff} presented at the beginning of Section~\ref{sec:CVPmodels} does not rely on the fact that $B_\zl$ is classical, but on such monogamy of quantum correlations.

In Section~\ref{sec:SPvsLP}, we will derive new relations between \SP and \LP tests valid for our general definition of short-range quantum correlations. Unlike the bound \eqref{eq:CHSHtradeoff}, such relations will not follow from the monogamy of quantum correlations. This is because routed Bell experiments should, in general, be considered as bipartite experiments where Bob's entire quantum system is routed either to one measurement device or the other measurement device, and cannot always be viewed as tripartite experiments by dividing Bob's system into a pair of subsystems, one for each measurement location.

\section{Routed Bell experiments and short-range quantum correlations} \label{sec:RoutedBellSRQ}
In this Section and the subsequent ones, we undertake a more formal and systematic analysis of correlations in routed Bell experiments. We consider a general routed Bell experiment as depicted in Fig.~\ref{fig:routedBell}. 
We introduce the following notation, which was already implicitly used in the previous sections. 

The state generated by the source is denoted as $\rho_{AB}$ and Alice's measurements are denoted as $M_{x}$ with POVM elements $M_{a|x}$. On Bob's side, the measurements are denoted as $M_{yz}$ with POVM elements $M_{b|yz}$; these operators depend not only on the local input $y$, but also on the switch value $z$, since the measurements made by the \SP device $B_{\zs}$ and the \LP device $B_{\zl}$ are not necessarily identical.
We assume that Alice has $m_{A}$ input choices and $d_{A}$ possible output results, i.e., $x\in\{0,1,\ldots,m_A-1\}$ and $a\in\{0,1,\ldots,d_A-1\}$. Similarly, Bob's measurement device on the short-path has $m_{B_\zs}$ inputs and $d_{B_\zs}$ outputs, while the device on the long-path has $m_{B_\zl}$ inputs and $d_{B_\zl}$ outputs. As before, the switch input is binary, taking values $z\in\{\zs,\zl\}$ to determine the routing of Bob's particle. 

The correlations in routed Bell experiments are characterized by the conditional probabilities $p=\{p(a,b|x,y,z)\}_{a,b,x,y,z}$. Note that, in any given run, a measurement is performed at Bob's side in only one of the two measurement locations. Hence the conditional probabilities involve a single input $y$ and a single output $b$, with the switch value $z$ indicating the corresponding location.

As in traditional Bell experiments, no-signalling constraints between Alice and Bob are assumed to hold. Specifically, the input and output of Alice should not causally influence Bob's outcome, and vice versa. In particular, the switch input $z$ should not have any causal influence on Alice. Such conditions can be enforced in a relativist framework by appropriately configuring the spacetime setup. In device-independent applications, it is commonly assumed that devices are internally described as black boxes, but are unable to transmit arbitrary external information. Relativistic constraints are typically not necessary in such scenarios. However, caution must be exercised in a routed Bell experiment, particularly in adversarial settings, to ensure that information about the switch input $z$ cannot be obtained after the switch operation has taken place. For instance, it is important to prevent situations where an adversary monitoring the transmission lines could determine wether a quantum particle has taken the short-path or long path and manipulate Alice's particle based on this knowledge before it reaches her measurement device.  Such actions would violate the no-signalling condition.

\subsection{Correlations in routed Bell experiments}
We now define various types of quantum correlations that can be observed in routed Bell experiments. We begin by considering the most general case, where no restrictions are imposed on the long path device $B_\zl$.

\subsubsection{General quantum correlations}
For a general quantum strategy, the correlations in a routed Bell experiment can be expressed as follows
\begin{equation}
    \label{eq:qcorr}
    p(a,b|x,y,z) = \Tr \left[(\mathcal{I}\otimes {C}_z)(\rho_{AB})\, M_{a|x}\otimes M_{b|yz}\right]\,,
\end{equation}
where ${C}_z$ is the CPTP map describing the transmission of Bob's system on the short-path ($z=\zs$) or the long-path ($z=\zl$). The quantum channel describing the transmission of the quantum system on Alice's side is independent of all the input variables $x,y,z$ and can thus be absorbed in the definition of the state $\rho_{AB}$.

The adjoint of the channels ${C}_z$ map the POVM elements $M_{b|yz}$ to valid POVM elements ${C}_z^\dagger (M_{b|yz})=\tilde M_{b|yz}$. Consequently, the above correlations can also be expressed as
\begin{equation}
    \label{eq:qcorr2}
p(a,b|x,y,z)= \Tr \left[\rho_{AB}\, M_{a|x}\otimes {C}_z^\dagger (M_{b|yz})\right] = \Tr \left[\rho_{AB}\, M_{a|x}\otimes \tilde M_{b|yz}\right]\,.
\end{equation}
Thus general correlations in a routed Bell experiment coincide with those of a regular bipartite Bell experiment where Bob has $m_{B_\zs}+m_{B_\zl}$ inputs represented by the pairs $(y,z)\in\{(0,\zs),\ldots,(m_{B_\zs}-1,\zs),(0,\zl),\ldots,(m_{B_\zl}-1,\zl)\}$. This simply expresses that the combined effect of the channel $C_z$ and the subsequent measurement $M_{yz}$ represents an effective measurement $\tilde M_{yz}$. This is illustrated in Fig.~\ref{fig:Routed_Bell_probs}.

\begin{figure}[t]
    \centering
    \includegraphics{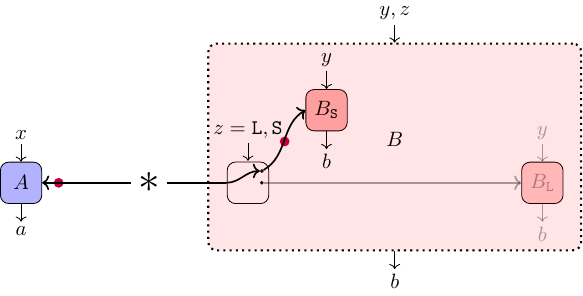}
    \caption{General correlations in a routed Bell experiment, $p(a,b|x,y,z)$ can be viewed as a regular bipartite Bell experiment where the combined effect of the switch and of the subsequent measurement (either $B_\zl$ or $B_\zs$) can be viewed as one effective measurement device for Bob, with input pairs $(y,z)\in\{(0,\zs),\ldots,(m_{B_\zs}-1,\zs),(0,\zl),\ldots,(m_{B_\zl}-1,\zl)\}$. \label{fig:Routed_Bell_probs}}
\end{figure}

We denote by $\Q$ the set of general quantum correlations \eqref{eq:qcorr} or \eqref{eq:qcorr2}. 

\subsubsection{Short-range quantum correlations} \label{sec:srq}
We now define the set of \emph{short-range quantum} (SRQ) correlations, denoted $\SRQ$, as the subset of the correlations \eqref{eq:qcorr} that can be obtained without any entanglement being distributed to $B_\zl$, i.e., as those where the channel $C_\zl$ in \eqref{eq:qcorr} is entanglement-breaking:
\begin{equation}
    \label{eq:srqcorr}
    p(a,b|x,y,z) = \Tr \left[(\mathcal{I}\otimes {C}_z)(\rho_{AB})\, M_{a|x}\otimes M_{b|yz}\right]\text{ with } C_{\zl} \text{ entanglement-breaking.}
\end{equation}

An entanglement-breaking channel $C_\zl$ can be understood as first performing a measurement described by POVM elements $N_\lambda$ on the input system, and then preparing one of the states $\{\rho_\lambda\}$ \cite{Pusey2015}:
\begin{equation}
    C_\zl(\rho) = \sum_\lambda \Tr\left[N_\lambda \rho\right]\,\rho_\lambda\,.
\end{equation}
The adjoint $C^\dagger_\zl$ maps the POVM elements $M_{b|y\zl}$ to
\begin{equation}\label{eq:jm}
\tilde M_{b|y\zl}=C^\dagger_\zl(M_{b|y\zl}) = \sum_\lambda p(b|y,\lambda) N_\lambda\,,  \quad \text{where } p(b|y,\lambda) = \Tr\left[\rho_\lambda\,M_{b|y\zl}\right]\,.
\end{equation}
This is equivalent to the statement that the measurements $\tilde M_{y\zl}$ defined by the operators $\{\tilde M_{b|y\zl}\}$ are jointly-measurable \cite{busch1986unsharp,lahti2003coexistence,heinosaari2016invitation}. That is, they can be reproduced by measuring a parent POVM $\{N_\lambda\}$, regardless of the input $y$, which yields a classical outcome $\lambda$. The final outcome $b$  is then generated according to the probability distribution $p(b|y,\lambda)$, which depends on both $y$ and $\lambda$.

As for the case of general quantum correlations, we can thus view SRQ correlations as bipartite Bell correlations with $m_{B_\zs}\times m_{B_\zl}$ inputs on Bob's side, but with the additional restriction that the subset of measurements corresponding to the input pairs $(y,\zl)$ are jointly-measurable, i.e., the correlations \eqref{eq:srqcorr} can also be expressed as
\begin{equation}
    \label{eq:srqcorr2}
p(a,b|x,y,z)=  \Tr \left[\rho_{AB}\, M_{a|x}\otimes \tilde M_{b|yz}\right]\text{ with the operators } \tilde M_{b|y\zl}\text{ jointly-measurable}\,.
\end{equation}

Using the definition of joint-measurability provided by \eqref{eq:jm}, we can explicitly express SRQ correlations as follows
\begin{equation}
    \label{eq:srqcorr3}
    p(a,b|x,y,z) =
        \begin{cases}
            \Tr \left[\rho_{AB}\, M_{a|x}\otimes \tilde M_{b|y\zs}\right] & \text{ if } z=\zs\\
            \sum_\lambda p(b|y,\lambda)\,\Tr \left[\rho_{AB}\, M_{a|x}\otimes N_{\lambda}\right] & \text{ if } z=\zl\,.
        \end{cases}
\end{equation}
Operationally, this can be understood as follows: if the switch selects the short path $(z=\zs)$, then the correlations are obtained by measuring a shared entangled state $\rho_{AB}$ as in a regular Bell experiment. If the switch selects the long path $(z=L)$, then a fixed measurement $\{N_\lambda\}$ is performed on Bob's system yielding a classical outcome $\lambda$. This classical outcome is then transmitted to $B_\zl$, which based on the input $y$ selects the output $b$ according to the distribution $p(b|y,\lambda)$. The example of Section~\ref{sec:example} clearly falls in this category. This aligns perfectly with our notion, formulated  in the Introduction, that SRQ correlations should not allow for the distribution of quantum information to $B_\zl$. Hence, the most general thing to do is to measure Bob's particle beyond the switch, and subsequently send a classical message to $B_\zl$.

The two equivalent definitions (\ref{eq:srqcorr}) and (\ref{eq:srqcorr2}) of SRQ correlations imply that they can be established without entanglement between $A$ and $B_\zl$, or, equivalently, when the device $B_\zl$ only performs compatible measurements. In contrast, long-range quantum (LRQ) correlations, defined as the set $\LRQ = \Q \backslash \SRQ$, certify both incompatible measurements at $A$ and $B_\zl$ and entanglement between them. This is in line with our discussion in Section \ref{CVPresources}.

\subsubsection{Fully quantum marginal correlations}
As mentioned in Section~\ref{sec:monogamy},  a subclass of general quantum correlations in routed Bell scenarios are those where the source prepares on Bob's side a pair of systems $B=(B_\zs, B_\zl)$ and the switch routes the first subsystem to the nearby device $B_\zs$ if $z=\zs$ and the second subsystem to $B_\zl$ if $z=\zl$. The resulting correlations are
\begin{equation}
    \label{eq:qqcorrelations}
    p(a,b|x,y,z) =
        \begin{cases}
             \Tr\left[\rho_{AB_\zs B_\zl}\, M_{a|x}\otimes M_{b|y,\zs}\otimes I\right] & \text{ if } z=\zs\\
            \Tr \left[\rho_{AB_\zs B_\zl}\, M_{a|x}\otimes I\otimes M_{b|y,\zl}\right] & \text{ if } z=\zl\,,
        \end{cases}
\end{equation}
and correspond to bipartite marginals of the tripartite $qqq$-correlations
\begin{equation}
    p(a,b_\zs,b_\zl|x,y_\zs,y_\zl) =
             \Tr\left[\rho_{AB_\zs B_\zl}\, M_{a|x}\otimes M_{b|y,\zs}\otimes M_{b|y,\zl}\right]\,.
\end{equation}
We refer to such correlations as $qq$-marginal correlations and denote the set of such correlations as $\QQ$.

\subsubsection{Quanutum-classical marginal correlations}
If we further restrict the state $\rho_{AB_\zs B_\zl}$ in the above correlations to be a $qqc$-state as in \eqref{eq:qqc}, i.e., \begin{equation}
    \label{eq:qqc2}
    \rho_{AB_\zs B_\zl} = \sum_{\lambda} p(\lambda)\, \rho_{AB_\zs}^\lambda \otimes |\lambda\rangle\langle \lambda|_{B_\zl}\,,
\end{equation}
then we get the class of correlations considered in \cite{Chaturvedi2022}, as pointed out in Section~\ref{sec:monogamy}. From now on, we will refer to such correlations as $qc$-marginal correlations and denote the corresponding set as $\QC$.
\subsubsection{Relations between the above correlations}
The following figure depicts the relation between the sets of correlations defined above.
\begin{equation}
   \includegraphics{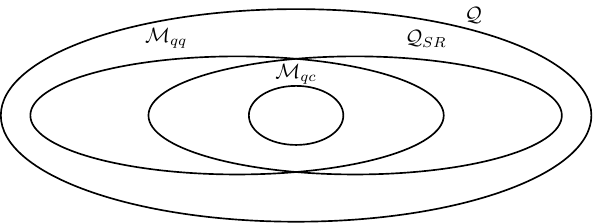}   
\end{equation}

The inclusions $\QC\subseteq \QQ \subseteq \Q$ and $\QC\subseteq \SRQ\subseteq \Q$ are obvious. They are actually strict, i.e., $\QC\subset\QQ \subset \Q$ and $\QC\subset \SRQ\subset \Q$. That $\QC \neq \QQ$ and $\SRQ \neq \Q$ is easy to see. The example of Section~\ref{sec:example} and the fact that $\QC$ and $\QQ$ satisfy the bound \eqref{eq:CHSHtradeoff}, as pointed out in Section~\ref{sec:monogamy}, imply $\QC\neq \SRQ$ and $\QQ \neq \Q$. 

It is further not difficult to see that $\QQ$ and $\SRQ$ are incomparable, meaning that there exist correlations that belong to one set but not the other, and vice versa. Indeed, the example of Section~\ref{sec:example} shows that $\SRQ \nsubseteq \QQ$, while $\QQ \nsubseteq \SRQ$ follows for instance from the fact that correlations in $\QQ$ can yield a \LP CHSH value $\chsh_\zl>2$, while correlations in $\SRQ$ cannot.

\subsection{Characterization through semidefinite programming hierarchies} \label{sec:SDPcharacterization}
All of the above sets can be outer-approximated through semidefinite programming (SDP) hierarchies for noncommutative polynomial optimization \cite{NPA2007,NPA08,PNA10}. This is because we can express the different types of correlations $p$ above as 
\begin{equation}
    p(a,b|x,y,z) = \Tr\left[\rho\, M_{a|x}\, M_{b|yz}\right]
\end{equation}
where the $d_A\times m_{A}$ measurement operators $M_{a|x}$ and the $d_{B_\zs}\times m_{B_\zs}+d_{B_\zl}\times m_{B_\zl}$ measurement operators $M_{b|y,z}$ are projectors and satisfy specific commutation relations depending on the type of correlations $p$:
\begin{align}
    &[M_{a|x},M_{b|yz}]= 0 && \text{ if } p\in\Q\,,&\\
    &[M_{a|x},M_{b|yz}]= 0, \, [M_{b|y\zl},M_{b'|y'\zl}]=0 && \text{ if } p\in\SRQ\,,& \label{eq:SRQ-sdp-characterization}\\
    &[M_{a|x},M_{b|yz}]= 0, \, [M_{b|y\zs},M_{b'|y'\zl}]=0 && \text{ if } p\in\QQ\,,&\\
    &[M_{a|x},M_{b|yz}]= 0, \, [M_{b|y\zs},M_{b'|y'\zl}]=0,\,[M_{b|y\zl},M_{b'|y'\zl}]=0  && \text{ if } p\in\QC\,.&
\end{align}
These representations naturally fit in the framework of noncommutative polynomial optimization. 

The above follow from the fact that in each of the representations \eqref{eq:qcorr2}, \eqref{eq:srqcorr2}, \eqref{eq:qqcorrelations}, the tensor product structure between different subsystems can be replaced by commutation relations. Specifically, the tensor product structure between $A$ and $B$, common to all types of correlations, can be replaced by the commutations relations $[M_{a|x},M_{b|yz}]= 0$. In the case of $\QQ$ and $\QC$, the additional product structure between $B_\zs$ and $B_\zl$ leads to the commutation relations $[M_{b|y\zs},M_{b'|y'\zl}]=0$. Furthermore, we can without loss of generality assume the measurements to be projective by dilating the local Hilbert spaces if necessary. The requirement of joint-measurability in \eqref{eq:srqcorr2} to define $\SRQ$ is then equivalent to the condition that the operators $M_{b|y\zl}$ commute with each other \cite{Uola2014}, i.e., to $[M_{b|y\zl},M_{b'|y'\zl}]=0$. Lastly, the condition that the subsystem $B_\zl$ is classical in the case of $\QC$ is equivalent to the condition that the operators $M_{b|y\zl}$ commute with all other operators, leading also to the additional commutation relations $[M_{b|y\zl},M_{b'|y'\zl}]=0$ in that case.
Note that, strictly speaking, the replacement of the tensor product structure by commutation relations represents a relaxation for infinite-dimensional quantum systems \cite{Slofstra2019,Zhengfen2022}. However, any commuting infinite-dimensional quantum correlations that are described by our current physical theories can be approximated arbitrarily well by tensor-product quantum correlations.

In the case of $\SRQ$, an alternative characterization is possible \cite{Pauwels2022}. Indeed, in \eqref{eq:srqcorr3}, we can assume without loss of generality that the outcome $\lambda$ of the measurement $\{N_\lambda\}$ specifies deterministically the outcome $b$ for each input $y$ of $B_\zl$. That is, we can assume that $\lambda = \boldsymbol{\beta} = (\beta_0,\ldots,\beta_{m_{B_\zl}-1})$, where $\beta_y$ is the output for input $y\in\{0,\ldots,m_{B_\zl}-1\}$. We can then write SRQ correlations as
\begin{equation} \label{eq:srqcorr4}
    p(a,b|x,y,z) =
        \begin{cases}
            \Tr \left[\rho_{AB}\, M_{a|x}\otimes M_{b|y\zs}\right] & \text{ if } z=\zs\\
            \sum_{\boldsymbol{\beta}} \delta_{\boldsymbol{\beta}_y,b} \,\Tr\left[\rho_{AB}\, M_{a|x}\otimes N_{\mathbf{\boldsymbol{\beta}}}\right]
& \text{ if } z=\zl\,.
        \end{cases}
\end{equation}
In other words, SRQ correlations are linear combinations of regular bipartite quantum correlations that involve on Bob's side $m_{B_\zs}$ measurements with $d_{B_\zs}$ outcomes, corresponding to the operators $M_{b|y\zs}$, and one additional measurement with $d_{B_\zl}^{m_{B_\zl}}$ outcomes, corresponding to the operators $N_{\boldsymbol{\beta}}$. They can thus be approximated from the outside, like regular bipartite quantum correlations, using the SDP hierarchies \cite{NPA2007,NPA08,PNA10}.

\section{\emph{SP}-enhancement of \LP tests} \label{sec:SPvsLP}
As the example of Section~\ref{sec:example} shows, the short-path CHSH value $\chsh_\zs$ does not constrain the long-path CHSH value $\chsh_\zl$, according to our definition of SRQ correlations. We show in this section, though, that there exist other \LP tests for which a \SP CHSH violation does weaken the conditions under which they witness long-range quantum correlations. We refer to this as a ``\emph{SP}-enhancement'' of the \LP test.

Throughout this section, we assume that all inputs and outputs of the measurement devices are binary. We denote the input values as $x,y\in\{0,1\}$, and for convenience the output values as $a,b\in\{\pm 1\}$. 
We define $A_x=\sum_{a=\pm 1} a\, M_{a|x}$ as the observable corresponding to the average value of $a$ for given input $x$ and we define similarly $B_{yz}$ based on the effective POVMs $\tilde M_{b|yz}$ that appear in definition \eqref{eq:srqcorr2}. In the case where the POVM elements $M_{a|x}$ and $\tilde M_{b|yz}$ are projective (which we can assume without loss of generality), the observables $A_x$, $B_{yz}$ are unitary and square to the identity. 

We use these observables to define the observed quantities $\langle A_x\rangle = \Trace{\left[\rho_{AB}\,A_x\right]} = \sum_{a=\pm 1} a\,p(a|x)$, $\langle B_{yz}\rangle = \Trace{\left[ \rho_{AB}\,B_{yz}\right]} = \sum_{b=\pm1} b\,p(b|y,z)$, and $\langle A_x B_{yz}\rangle = \Trace{\left[ \rho_{AB}\,A_xB_{yz}\right]} = \sum_{a,b=\pm 1}ab ~ p(a,b|x,y,z)$. 
The knowledge of $\langle A_x\rangle$, $\langle B_{yz}\rangle$, and $\langle A_xB_{yz}\rangle$ is equivalent to the knowledge of the full set of probabilities $p(a,b|x,y,z)$.

\subsection{A family of \LP tests}
The \LP tests we are going to consider in this section are based on the following Bell expressions 
\begin{align}
    \mathcal{J}^{\theta}_{\zl} = t_\theta \langle A_0 B_{0\zl} \rangle +  \langle A_0 B_{1\zl} \rangle +  \langle A_1 B_{0\zl} \rangle - t_\theta\langle A_1 B_{1\zl} \rangle\,,  \label{eq:SimpleJM-ineq-bipartite}
\end{align}
where $t_\theta=\tan\theta$ and $\theta$ is a parameter in $\left[0,\pi/4\right]$. The expressions $\mathcal{J}^{\theta}_{\zl}$ for other possible values of $\theta$ can be obtained by relabelling the input and/or outputs of the observables $A_x$ and $B_{y\zl}$. 

Seen as standard Bell expressions, $\mathcal{J}^{\theta}_{\zl}$ satisfy the following local and quantum bounds (see Appendix~\ref{app:SimpleJM-bound-proof}):
\begin{align}
    \mathcal{J}^{\theta}_{\zl} &\leq 2 \qquad\text{(local bound)} ~,  \label{eq:SimpleJM-localbound}\\
    \mathcal{J}^{\theta}_{\zl} &\leq 2/c_\theta \qquad\text{(quantum bound)}\,, \label{eq:SimpleJM-qmbound}
\end{align}
where $c_\theta = \cos\theta$.
The following states and observables
\begin{equation}\label{eq:strat}
    \rho_{AB} = \ketbra{\phi_+}\,, \quad A_0 = X\,,\,A_1=Z\,,\quad B_{0\zl} = s_\theta X + c_\theta Z \,,\, B_{1\zl} = c_\theta X - s_\theta Z\,,
\end{equation}
define the optimal quantum strategies reaching the maximal quantum bound $2/c_\theta$. Note that the measurements on Bob's side are anticommuting and the angle $\theta$ can be seen as a global rotation along the $Y$ axis on the Bloch sphere.

For $\theta=\pi/4$, the expression \eqref{eq:SimpleJM-ineq-bipartite} simply corresponds to the CHSH expression $ \mathcal{J}^{\pi/4}_{\zl}= \mathcal{C}_\zl$ and the local bound \eqref{eq:SimpleJM-localbound} and quantum bounds \eqref{eq:SimpleJM-qmbound} are the usual ones, i.e., $2$ and $2\sqrt{2}$. For $\theta=0$ the local and quantum bounds are both equal to 2, i.e., $\mathcal{J}^0_\zl$ cannot be used to detect any quantum nonlocality. Values of $\theta$ between $0$ and $\pi/4$ lead to a gap between the local and quantum bounds, i.e., the inequalities $\mathcal{J}^{\theta}_{\zl}\leq 2$ correspond to standard Bell inequalities.

In the next two subsections, we will see how the above family of \LP tests can be enhanced by a \SP CHSH test. 

\subsection{\emph{SP}-enhancement with a maximal short-path CHSH value}\label{subsec:SimpleJM}
Assume that in addition to the \LP expression $\mathcal{J}^\theta_\zl$, we observe a \SP CHSH value, which, as a starting point, we assume to be maximal, i.e., $\chsh_\zs=2\sqrt{2}$.
\begin{prop}\label{prop1}
    When $\mathcal{C}_\zs = 2\sqrt{2}$, SRQ correlations satisfy the following bound 
    \begin{equation}
        \mathcal{J}^{\theta}_{\zl} \leq  \sqrt{2}/{c_\theta} \qquad\text{(SRQ bound)} \label{eq:SimpleJM-srqbound}\\\,.
    \end{equation}
\end{prop}
\begin{proof}
    To prove this bound, we need to maximize the \LP expression $\mathcal{J}^\theta_\zl$ for SRQ correlations 
    of the form \eqref{eq:srqcorr2}. The fact that $\mathcal{C}_\zs = 2\sqrt{2}$ fixes the shared state $\rho_{AB}$ and the observables $A_x$ and $B_{y\zs}$. Indeed by self-testing \cite{Cirelson1980, Mayers2004, Supic2020}, they must be, up to local isometries, equivalent to the two-qubit state $|\phi_+\rangle = (\ket{00} + \ket{11})/\sqrt{2}$ and the qubit observables $A_0 = X$, $A_1=Z$, $B_{0\zs} = (Z+X)/\sqrt{2}$, $B_{1\zs} = (X-Z)/\sqrt{2}$. Consequently, the \LP correlators are equal to
    \begin{align}\label{eq:pmcorrelators}
        \langle A_0 B_{y\zl}\rangle &= \langle \phi_+|X \otimes B_{y\zl}|\phi_+\rangle =\frac{1}{2}\Trace{\left[XB_{y\zl}\right]}\,,\\
        \langle A_1 B_{y\zl}\rangle &= \langle \phi_+|Z \otimes B_{y\zl}|\phi_+\rangle =\frac{1}{2}\Trace{\left[ZB_{y\zl}\right]}\,.\nonumber
    \end{align}
    We can thus write the \LP expression $\mathcal{J}^\theta_\zl$ as
    \begin{align}
        \mathcal{J}^\theta_\zl = \frac{1}{2}\left[t_\theta \Tr(X B_{0\zl}) + \Tr(X B_{1\zl}) + \Tr(Z B_{0\zl}) -t_\theta \Tr(Z B_{1\zl})\right] \, . \label{eq:Jtheta} 
    \end{align}
    We now need to bound the above expressions for observables $B_{0\zl}$ and $B_{1\zl}$ that are jointly-measurable. In the case $\theta=0$, it is shown in \cite{Pusey2015} that
    \begin{equation}
        \frac{1}{2} \left[ \Trace(XB_{1\zl}) + \Trace(ZB_{0\zl}) \right] \le \sqrt{2} \,, \label{eq:SimpleJM-ineq}
        \end{equation}
    whenever $B_{0\zl}$ and $B_{1\zl}$ are jointly-measurable. We rederive this result in Appendix~\ref{app:SimpleJM-proof} for completeness. The above joint-measurability inequality holds for any pair of observables $B_{0\zl}$ and $B_{1\zl}$ and is independent of the basis in which we write it. Making the change of basis $Z\rightarrow s_\theta X + c_\theta Z$, $X\rightarrow c_\theta X - s_\theta Z$, where $c_\theta = \cos{\theta}$ and $s_\theta = \sin{\theta}$, we can rewrite the above inequality as
    \begin{align}
        \frac{1}{2} \left[ s_\theta\Tr(X B_{0\zl}) + c_\theta\Tr(X B_{1\zl}) + c_\theta\Tr(ZB_{0\zl}) - s_\theta \Tr(Z B_{1\zl}) \right]\leq \sqrt{2}\,. \label{eq:SimpleJM-ineq-2}
    \end{align}
    Dividing by $c_\theta$, the left-hand side becomes equal to \eqref{eq:Jtheta} and the right-hand side to $\sqrt{2}/c_\theta$, proving \eqref{eq:SimpleJM-srqbound}.
\end{proof}

The intuition behind the above Proposition and its proof is that when Alice and Bob observe the maximal value $\chsh_\zs=2\sqrt{2}$, their \LP correlators are, by self-testing, directly associated to the Pauli expectations $\Tr{\left[PB_{y\zl}\right]}$, where $P={I},X,Z$, of the observables $B_{y\zl}$. This means that they can perform tomography, restricted to the $Z\shortminus X$ plane, of these observables.
If from the knowledge of these Pauli expectations it is possible to conclude that the two observables $B_{0\zl}$ and $B_{1\zl}$ are not jointly measurable, then, following the definition of SRQ correlations in Subsection~\ref{sec:srq}, the resulting correlations are not SRQ. Inequalities that rule out SRQ correlations are thus directly related to inequalities that rule out joint-measurability, such as \eqref{eq:SimpleJM-ineq} and \eqref{eq:SimpleJM-ineq-2}. In fact, it is the existence of these inequalities for joint-measurability that led us to introduce the \LP expressions \eqref{eq:SimpleJM-ineq-bipartite}.

The SRQ bound \eqref{eq:SimpleJM-srqbound}, the local bound \eqref{eq:SimpleJM-localbound}, and the general quantum bound \eqref{eq:SimpleJM-qmbound} of 
$\mathcal{J}^{\theta}_{\zl}$ are plotted in Fig.~\ref{fig:bounds_simple_JM} as a function $\theta$. 
For any value $0\leq\theta<\pi/4$, the SRQ bound is strictly smaller than the local bound, i.e., the \SP CHSH test weakens the condition to witness long-range quantum correlations based on the $\mathcal{J}^{\theta}_{\zl}$ test. When $\theta=\pi/4$, the expression $\mathcal{J}_{\zl}^{\pi/4}$ is simply the CHSH expression $\chsh_\zl$. In this case, the SRQ bound and local bound coincide, i.e., the \SP CHSH test does not weaken the condition to witness long-range quantum correlations using the $\chsh_\zl$ test, as already noted in Subsection~\ref{sec:example}.
\begin{figure}[t!]
    \centering
    \includegraphics{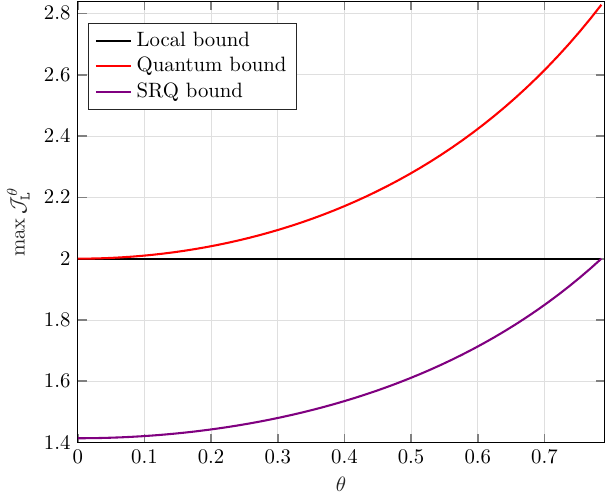}
    \caption{The SRQ bound (assuming $\chsh_\zs=2\sqrt{2}$), the local bound, and the general quantum bound of 
    $\mathcal{J}^{\theta}_{\zl}$ as a function of $\theta$.}
    \label{fig:bounds_simple_JM}
\end{figure}

Note that by adding the short-path measurements 
\begin{equation}\label{eq:stratSP}
    B_{0\zs} = \frac{X+Z}{\sqrt{2}}, \quad B_{1\zs} = \frac{X-Z}{\sqrt{2}}\,,
\end{equation}
to the strategy \eqref{eq:strat}, both $\mathcal{C}_\zs$ and $\mathcal{J}^\theta_{\zl}$ can simultaneously reach their maximal quantum values of $2\sqrt{2}$ and $2/c_\theta$, respectively.

Interestingly, in the case $\theta=0$, the \emph{SP}-enhancement of the \LP test based on $\mathcal{J}_{\zl}\equiv \mathcal{J}^0_{\zl}$ is maximal. Indeed, in this case the standard local bound \eqref{eq:SimpleJM-localbound} is
\begin{equation}
    \mathcal{J}_\zl = \langle A_0 B_{1\zl}\rangle + \langle A_1 B_{0\zl}\rangle \leq 2\,.
\end{equation}
But this coincides with the quantum bound \eqref{eq:SimpleJM-qmbound} and thus the above inequality cannot be violated by quantum theory. That is, it does not represent a proper Bell inequality and cannot be used to witness long-range quantum nonlocality. However, when supplemented with a maximal \SP test, the local bound has to be replaced with the more strict SRQ bound
\begin{equation}
    \mathcal{J}_\zl = \langle A_0 B_{1\zl}\rangle + \langle A_1 B_{0\zl}\rangle  \leq \sqrt{2} < 2\,.
\end{equation}
This bound is now smaller than the quantum bound, hence it does represent a proper witness of long-range quantum nonlocality. The strategy that reaches the maximal quantum values $\mathcal{C}_\zs=2\sqrt{2}$ and $\mathcal{J}_\zl=2$ is defined by eqs.~\eqref{eq:strat} and \eqref{eq:stratSP}. It is obtained by measuring a maximally two-qubit state $|\phi_+\rangle$ state with $Z,X$ observables for $A$ and the remote devices $B_\zl$. These correlations arise in many contexts in quantum information -- e.g. in entanglement-based BB84 \cite{Bennet1984,Bennet1992,Shor2000}, or for witnessing the entanglement of the $|\phi_+\rangle$ state \cite{Yu2005,Guhne2009} -- but they cannot certify any quantum property in a device-independent way since they can be reproduced by purely classical strategies. The result above shows that if we append to this BB84 scenario intermediate $B_\zs$ measurements in the $(X\pm Z)/\sqrt{2}$ basis, then the $\langle {A}{B_\zl}\rangle$ correlations certify long-range quantumness in a device-independent setting.

To illustrate the interest of the \emph{SP}-enhanced \LP tests represented by the family of inequalities \eqref{eq:SimpleJM-srqbound}, consider the optimal correlations defined by eqs.~\eqref{eq:strat} and \eqref{eq:stratSP}, but assume that the quantum particle going to the remote measurement device $B_\zl$ is characterized by a visibility $\nu$, i.e., with probability $1-\nu$ it undergoes completely depolarizing noise. Then the corresponding long-range quantum correlations are 
\begin{equation}\label{eq:Simplecorr}
    \begin{aligned}
        & \langle A_0 B_{1\zl}\rangle = \langle A_1 B_{0\zl}\rangle = \nu c_\theta,\quad \langle A_0 B_{0\zl}\rangle = -\langle A_1 B_{1\zl}\rangle = \nu s_\theta\,,
    \end{aligned}
\end{equation}

Without the switch, the condition to witness non-locality using the $\mathcal{J}^\theta_{\zl}$ expression is $\mathcal{J}^\theta_{\zl}>2$, i.e., for the correlations \eqref{eq:Simplecorr} $2\nu/c_\theta > 2$. That is, a visibility $\nu>c_\theta$ is required. However, we know that in the two-input, two-output regular Bell scenario, a necessary and sufficient condition to witness non-locality is the violation $\chsh>2$ of the CHSH inequality. In the case of the correlations \eqref{eq:Simplecorr}, this condition is equivalent to $2\nu(c_\theta+s_\theta)>2$, or $\nu>1/(c_\theta+s_\theta)$. The required visibility threshold thus goes from $\nu>1$ when $\theta=0$ (i.e., the correlations do not exhibit any nonlocality) to $\nu>1/\sqrt{2}\simeq 0.71$ when $\theta=\pi/4$ (corresponding to the case of maximally robust Tsirelson's correlations). 

In a routed Bell experiment, and since the \SP CHSH test is maximally violated, we can witness long-range quantum correlations whenever $\mathcal{J}^\theta_{\zl}>\sqrt{2}/c_\theta$ for the \LP test, instead of the more constraining criterion $\mathcal{J}^\theta>2$. This happens when $2\nu/c_\theta>\sqrt{2}/c_\theta$, i.e., $\nu>1/\sqrt{2}\simeq 0.71$ for all $\theta$. Thus all correlations defined by the family of strategies \eqref{eq:strat} and \eqref{eq:stratSP} have the same noise tolerance in a routed Bell experiment, which is moreover equal to the noise tolerance of the maximally robust Tsirelson's correlations corresponding to $\theta=\pi/4$ in a regular Bell scenario.

\subsection{\emph{SP}-enhancement with a non-maximal short-path CHSH value}
Obviously, the assumption that $\chsh_\zs=2\sqrt{2}$ is too strong in any realistic experimental setting. We now derive 
new bounds on the value of $\mathcal{J}^{\theta}_{\zl}$ achievable by short-range quantum correlations when the CHSH value in the short-path is not maximal.    
We first consider the case of $\theta=0$, for which we analytically derive the SRQ bound for $\mathcal{J}_\zl \equiv \mathcal{J}_\zl^0$ as a function of $\chsh_\zs$.

\begin{prop}\label{thm:tradeoff-SRQvsLRQ-CHSHvsSimpleJ}
    For any short-range quantum (SRQ) correlations, the following inequality holds
    \begin{align}
        \mathcal{J}_\zl\leq \frac{\mathcal{C}_\zs + \sqrt{8-\mathcal{C}_\zs^2}}{2} \qquad\text{when }\mathcal{C}_\zs \in \left[2,2\sqrt{2}\right]\,. \label{eq:tradeoff-SRQvsLRQ-CHSHvsSimpleJ}
    \end{align} 
\end{prop}
\begin{proof}
In the $(\mathcal{C}_\zs,\mathcal{J}_\zl)$ plane, the region delimited by \eqref{eq:tradeoff-SRQvsLRQ-CHSHvsSimpleJ} is convex, hence it is equivalent to a family of linear bounds given by the tangent to the curve $\mathcal{J}_\zl= \frac{\mathcal{C}_\zs + \sqrt{8-\mathcal{C}_\zs^2}}{2}$. These linear bounds are
\begin{align}
   \sin u ~\mathcal{C}_\zs + (\cos u -\sin u)\mathcal{J}_\zl \leq 2 ~,~~ u \in \left[0,\frac{\pi}{4}\right]~. \label{eq:tradeoff-SRQvsLRQ-CHSHvsSimpleJ-2}
\end{align}
Instead of directly proving \eqref{eq:tradeoff-SRQvsLRQ-CHSHvsSimpleJ}, we will thus prove the bounds \eqref{eq:tradeoff-SRQvsLRQ-CHSHvsSimpleJ-2} for every $u\in[0,\pi/4]$. 

Let us define the Bell operators
\begin{align}
    \mathtt{C}_\zs & = A_0B_{0\zs} + A_0B_{1\zs} + A_1B_{0\zs} - A_1B_{1\zs}\\
    \mathtt{J}_\zl & = A_0B_{1\zl} + A_1B_{0\zl}\,.
\end{align}
Then since for any quantum state $\ket{\psi}$, $\mathcal{C}_\zs = \bra{\psi}\mathtt{C}_\zs\ket{\psi}$ and $\mathcal{J}_\zl = \bra{\psi}\mathtt{J}_\zl\ket{\psi}$,
proving \eqref{eq:tradeoff-SRQvsLRQ-CHSHvsSimpleJ-2} is equivalent to proving the operator semidefinite constraint
\begin{equation}
    \mathtt{I}_u \vcentcolon= 2 \mathbf{1} - \sin u ~\mathtt{C}_\zs - (\cos u -\sin u)\mathtt{J}_\zl \geq 0 \qquad\text{for all }u\in\left[0,\frac{\pi}{4}\right]\,. \label{eq:tradeoff-SRQvsLRQ-CHSHvsSimpleJ-3}
\end{equation}
We will prove this positivity constraint using a Sum of Squares (SoS) decomposition.

Define the hermitian operators:
\begin{subequations}
    \begin{align}
        \mathtt{C}_{ij \zs} &= \sum_{x,y} (-1)^{\delta_{x,i}\delta_{y,j}}A_xB_{y\zs} ~, \\
\mathtt{J}_{ij \zl} & =A_0B_{(0\oplus i)\zl} + (-1)^{j+1} A_1B_{(1\oplus i)\zl},
    \end{align}
\end{subequations}
for $i,j\in\{0,1\}$, where $\oplus$ denotes addition modulo 2. Thus $\mathtt{C}_\zs = \mathtt{C}_{11\zs}$ and $\mathtt{J}_\zl = \mathtt{J}_{11\zl}$.

Define further
\begin{subequations}\label{eq:SOS-basis}
    \begin{align} 
        P_1(u) &= -c_u(c_u-s_u)\mathtt{C}_{11 \zs} + (c_u^2-s_u^2)\mathtt{J}_{11 \zl} ~, \\
        P_2(u) &= -c_u(c_u+s_u)\mathtt{C}_{01 \zs} + (c_u^2-s_u^2)\mathtt{J}_{01 \zl} ~,\\
        P_3(u) &= s_u(c_u+s_u)\mathtt{C}_{10 \zs} + (c_u^2-s_u^2)\mathtt{J}_{10 \zl} ~,\\
        P_4(u) &= s_u(c_u-s_u)\mathtt{C}_{00 \zs} + (c_u^2-s_u^2)\mathtt{J}_{00 \zl}~ ,
    \end{align}
\end{subequations}
where $c_u=\cos{u}$ and $s_u=\sin{u}$. Then using that the operators $A_x$, $B_{yz}$ obey the commutation relations in \eqref{eq:SRQ-sdp-characterization} and that they square to the identity, it is easily verified that
\begin{align}
    \mathtt{I}_u = \frac{1}{4} \mathtt{I}_u^2 + \frac{s_uP_1(u)^2}{8 c_u(c_u^2 - s_u^2)} + \frac{s_uP_2(u)^2}{8 c_u(c_u + s_u)^2} + \frac{P_3(u)^2}{8 (c_u + s_u)^2} + \frac{P_4(u)^2}{8 (c_u^2 - s_u^2)} ,\quad \text{for }u \in [0,\frac{\pi}{4}[  ~. \label{eq:SOS-proof}
\end{align}
The right-hand side of the above expression is a SoS, hence is positive, which proves \eqref{eq:tradeoff-SRQvsLRQ-CHSHvsSimpleJ-3}. 
The SoS \eqref{eq:SOS-proof} is not valid at $u = {\pi}/{4}$. But for this point \eqref{eq:tradeoff-SRQvsLRQ-CHSHvsSimpleJ-2} is simply the well-known bound $\mathcal{C}_\zs\leq 2\sqrt{2}$ for CHSH. 

To find the above SoS decomposition, we followed the approach in \cite{Bamps2015}. We briefly describe this method in Appendix \ref{app:SOS-proof-methods}, where we also show that the bound \eqref{eq:tradeoff-SRQvsLRQ-CHSHvsSimpleJ-2} is tight.
\end{proof}

For other values of $\theta$, a SRQ bound on $\mathcal{J}^\theta_{\zl}$ as a function of $\chsh_\zs$ can be obtained numerically. Maximizing $\mathcal{J}_\zl^\theta$ for 
a fixed value of $\chsh_\zs$ over $\mathcal{Q}^n_{SR}$, the set of correlations corresponding to the $n^{\rm th}$-level NPA relaxation 
of the set of SRQ correlations \cite{NPA2007,NPA08,PNA10} (see Section \ref{sec:SDPcharacterization}), gives an upper bound on the maximum value of 
$\mathcal{J}_\zl^\theta$. Lower bounds on the maximum value of $\mathcal{J}_\zl^\theta$ can be obtained using 
a see-saw algorithm searching over explicit quantum strategies. 
However, fixing the value of $\chsh_\zs$ is not convenient when running the see-saw algorithm since we initialize on a random state and measurements. Instead, we 
maximize the expressions $\cos u ~\mathcal{C}_{\zs} + \sin u~ \mathcal{J}^{\theta}_{\zl}$ for different values of $u$, yielding the solution $k_u$. The envelope of the curves
$\cos u ~ \mathcal{C}_{\zs} + \sin u ~ \mathcal{J}^{\theta}_{\zl} = k_u$ yields lower-bounds on $\mathcal{J}^\theta_{\zl}$ for given $\chsh_\zs$.

A plot of such bounds for different values of $\theta$ is given in 
Fig. \ref{fig:tradeoff-SRQvsLRQ-CHSHvsSimpleJ}. The upper bounds were obtained using level `$AB$' of the NPA hierarchy.
The upper bounds match the see-saw lower bounds up to numerical precision.
\begin{figure}[t!]
    \centering
   \includegraphics{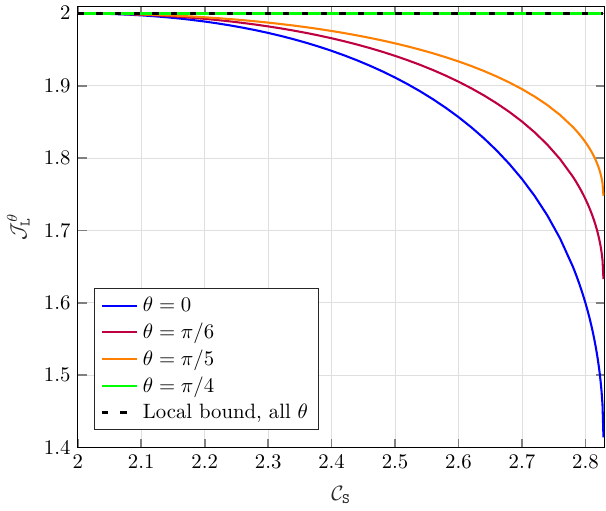}
    \caption{$\mathcal{J}_\zl^\theta$ vs $\mathcal{C}_\zs$ for SRQ strategies, obtained at level `$AB$' of the NPA hierarchy and saturated using a see-saw algorithm.
    \label{fig:tradeoff-SRQvsLRQ-CHSHvsSimpleJ}}
\end{figure}

As before, we illustrate the interest of such bounds on the noise robustness of the correlations induced by the quantum strategies \eqref{eq:strat} and \eqref{eq:stratSP}. We assume this time a depolarizing noise characterized by a visibility $\nu_\zs$ on the path from the source to the devices $A$ and $B_\zs$, which are located close to the source, and a visibility $\nu_\zl$ on the path to the remote device $B_\zl$, with $\nu_\zl\leq \nu_\zs$. The resulting correlations are
\begin{equation}\label{eq:Simpleecorr2}
    \begin{aligned}
        &\langle A_0 B_{0\zs}\rangle = \langle A_0 B_{1\zs}\rangle = \langle A_1 B_{0\zs}\rangle = - \langle A_1 B_{1\zs}\rangle  = \nu_\zs^2 \frac{\sqrt{2}}{2}\,,\\
        & \langle A_0 B_{1\zl}\rangle = \langle A_1 B_{0\zl}\rangle = \nu_\zs \nu_\zl c_\theta,\quad \langle A_0 B_{0\zl}\rangle = -\langle A_1 B_{1\zl}\rangle = \nu_\zs \nu_\zl s_\theta\,,
    \end{aligned}
\end{equation}
yielding $\chsh_\zs = \nu_\zs^2 2\sqrt{2}$, $\chsh_\zl = 2 \nu_\zs \nu_\zl (c_\theta + s_\theta)$, and $\mathcal{J}^\theta_\zl = 2  \nu_\zs \nu_\zl /c_\theta$. 

In a regular Bell scenario, the condition for the correlations \eqref{eq:Simpleecorr2} to exhibit non-locality is the violation of the CHSH inequality $\chsh_\zl>2$, i.e., $\nu_\zl>1/(\nu_\zs (c_\theta+s_\theta))$. The conditions obtained in a routed Bell scenario exploiting the above bounds on $\mathcal{J}^\theta_\zl$ induced by the $\chsh_\zs$ value are depicted in Fig.~\ref{fig:visibility_simpleJ} for different values of $\theta$.

These conditions are based on a specific family of \LP tests, i.e., \eqref{eq:SimpleJM-ineq-bipartite}, which are not necessarily optimal. We thus also directly determined the conditions on $\nu_\zs$ and $\nu_\zl$ required to demonstrate long-range quantum correlations using the full set of correlations \eqref{eq:Simpleecorr2} and the NPA relaxation of the SRQ set at a level intermediate between 3 and 4. These results are depicted in Fig.~\ref{fig:visibility_fix_corr}. We notice that for very high values of $\nu_\zs$, routed Bell experiments based on a \SP CHSH test tolerate lower values of $\nu_\zl$ than standard Bell experiments. The range of values of $\nu_\zs$ for which this \emph{SP}-enhancement is obtained increases as $\theta\rightarrow 0$. 

We note that for the BB84 correlations corresponding to $\theta=0$, using the LP inequality \eqref{eq:tradeoff-SRQvsLRQ-CHSHvsSimpleJ} or the full set of correlations give the same results, hence this inequality seems to be optimal for these correlations. The inequality \eqref{eq:tradeoff-SRQvsLRQ-CHSHvsSimpleJ} is violated by the noisy BB84 correlations whenever $2\nu_\zs \nu_\zl> \sqrt{2}\nu_\zs^2 +\sqrt{2(1-\nu_\zs^4)}$. The minimal value $\nu_\zs$ for which this arises, assuming further $\nu_\zs\geq \nu_\zl$ (since the \SP visibility is higher than the \LP visibility), is $\nu_\zs = 1/(4 - 2\sqrt{2})^{1/4}\simeq 0.9612$.

\begin{figure}[t!]
    \centering
    \begin{subfigure}[b]{0.49\textwidth}
        \includegraphics[width=\textwidth]{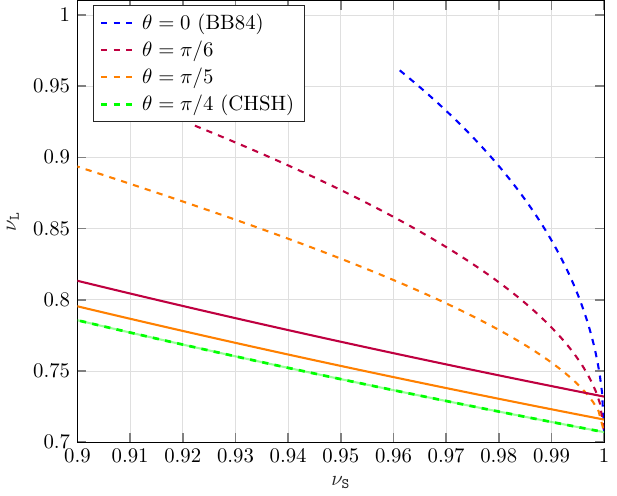}  
        \caption{}
        \label{fig:visibility_simpleJ}
    \end{subfigure}
    \hfill
    \begin{subfigure}[b]{0.49\textwidth}
        \includegraphics[width=\textwidth]{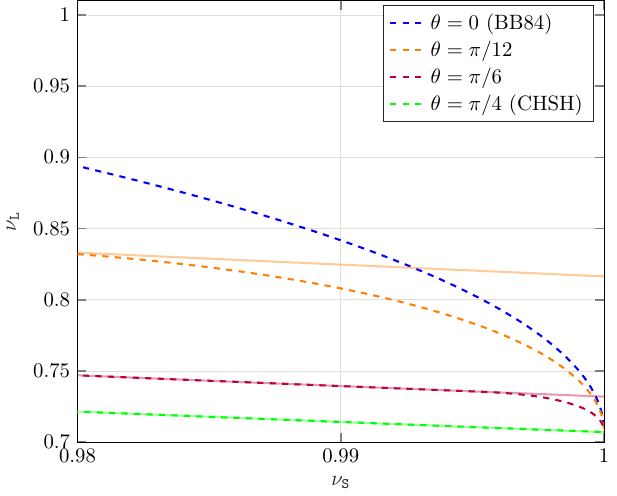}
        \caption{}
        \label{fig:visibility_fix_corr}
    \end{subfigure}
    \caption{Value of $\nu_\zl$ required to demonstrate long-range quantum correlations as a function of $\nu_\zs$ in a standard Bell scenario (solid lines) and routed Bell scenario (dashed lines) for the correlations \eqref{eq:Simpleecorr2}, determined (a) using the \LP test based on $\mathcal{J}^{\theta}_\zl$ and (b) using the full set of correlations \eqref{eq:Simpleecorr2} and the NPA relaxation of the SRQ set at level between 3 and 4.  }
    \label{fig:visibility} 
\end{figure}

\section{Detection efficiency thresholds in routed Bell experiments\label{sec:dec}}
We now consider the effect of detection inefficiencies in routed Bell experiments, the original motivation for considering such experiments.
We denote $\vec{\eta} = (\eta_A,\eta_{B_\zs},\eta_{B_\zl})$, the detector efficiencies of the three measurement devices, that is, the probabilities that these measurement devices `click' and provide a valid output. Disregarding the `no-click' events $\varnothing$ in the statistical analysis opens the so-called detection loophole, and is only valid assuming fair sampling \cite{Pearle1970,Clauser1974}.  There are two common ways to take no-click events into account: we can either count them as a separate, additional outcome $\varnothing$, or we can bin them with one of the other outcomes, say $\varnothing \mapsto +1$\footnote{We can also mix these two ways for different detectors.}. 

Given some correlations $p$ in the ideal situation where measurement devices always click, the non-ideal correlations $p^{\vec{\eta}}$ keeping the no-click outcomes $\varnothing$ as an additional outcome are
\begin{equation}\label{eq:pobs_nobinning}
    \begin{aligned}
        p^{\vec{\eta}}(a,b|x,y,z) &= \eta_{z}\eta_{A}p(a,b|x,y,z) \, , \\
       p^{\vec{\eta}}(a,\varnothing|x,y,z) &= \eta_{A}(1- \eta_{z})p(a|x) \, , \\
       p^{\vec{\eta}}(\varnothing,b|x,y,z) &=(1-\eta_{A})\eta_{z}p(b|y,z) \ , \\
        p^{\vec{\eta}}(\varnothing,\varnothing|x,y,z) &= (1-\eta_{A}) (1-\eta_{z}) \,,
    \end{aligned}
\end{equation}
where $p(a|x)$ and $p(b|y,z)$ are the marginal probabilities of respectively $A$ and $B_z$ in the target implementation.
If instead the no-click outcome is binned according to $\varnothing \mapsto +1$, the ideal correlations $p$ are modified to
\begin{multline} 
    p^{\vec{\eta}}(a,b|x,y,z) = \eta_{A}\eta_{z} p(a,b|x,y,z) + (1-\eta_{A})\eta_{z}p(b|y,z) \delta_{a,+1} \label{eq:pobs_binning} \\
   +  \eta_{A}(1- \eta_{z})  p(a|x) \delta_{b,+1} + (1-\eta_{A}) (1-\eta_{z}) \delta_{a,+1} \delta_{b,+1} \,
    ,
    \end{multline}
where $\delta$ is the Kronecker delta.

\subsection{Universal bounds on critical detection efficiency} \label{subsec:critical-efficiency-universal}

Lower bounds on critical detection efficiencies in regular Bell experiments 
were derived in \cite{Massar2003}. These bounds depend only on the number of measurement settings of Alice and Bob and are independent of the quantum implementation. We now generalize these bounds to i) the case where each party has a different detection efficiency (as \cite{Massar2003} considered only the symmetric case where Alice and Bob have the same detection efficiency) and ii) routed Bell experiments.
\begin{prop}\label{theo:universal-lb-switch}
    Consider a routed Bell experiment where Alice's measurement device has $m_A$ measurement settings and detection efficiency $\eta_A$, and Bob's remote device has $m_{B_\zl}$ measurement settings and detection efficiency $\eta_{B_\zl}$. Then, there exists an SRQ model when the following condition is satisfied,
    \begin{equation}
        \eta_{B_\zl} \le \frac{\eta_A(m_A - 1)}{\eta_A (m_A m_{B_\zl} - 1)- (m_{B_\zl} - 1)} \,, \label{eq:universal-bound-routed}
    \end{equation}
independently of the number of measurement settings $m_{B_\zs}$ and detection efficiency $\eta_{B_\zs}$ of Bob's close device $B_\zs$.
In particular, this bound also applies to standard Bell experiment, which can be seen as the particular case $m_{B_\zs}=0$.
\end{prop}
\begin{proof}
We prove the above result by constructing an explicit SRQ model for the correlations $p^{\vec{\eta}}$ given by \eqref{eq:pobs_nobinning} when \eqref{eq:universal-bound-routed} is satisfied. This model is based on a mixture of three different strategies with respective weights $st$, $s(1-t)$, and $1-s$ with $0\leq s,t\leq 1$. 

The first strategy is purely local (hence SRQ) and based on a hidden variable $(a',x')$ shared with all measurement devices and chosen with probability $p(a'|x')/m_A$ where $p(a'|x')$ is the marginal distribution of Alice's and Bob's ideal correlations $p$. If $x=x'$, Alice outputs $a'$, otherwise she outputs $\varnothing$. Bob's device (whether $B_\zs$ or $B_\zl$), on the other hand, outputs $b$ with probability $p(b|a',x',y) = p(a',b|x',y,z)/p(a'|x')$ when his input is $y$.

The second strategy is a SRQ strategy in which the source, Alice's measurement device $A$ and Bob's nearby measurement device $B_\zs$ are as in the quantum strategy yielding the ideal correlations $p$. On Bob's side, if the switch is set to $z=\zl$, then an input $y'$ is selected at random with probability $1/m_{{B_\zl}}$ and the corresponding ideal measurement (the one leading to the correlations $p$) is performed yielding the outcome $b'$. Both $y'$ and $b'$ are relayed to $B_{\zl}$ as a classical message.  If $B_\zl$'s actual input $y$ matches $y'$, it outputs $b=b'$, otherwise, it outputs $\varnothing$. 
Note that if $m_{B_\zs}=0$, i.e, there is no switch and no nearby measurement device, then this strategy is actually equivalent to a purely local one, since it involves a single quantum measurement on Bob's side.

In the third strategy, $A$ and $B_\zl$ always output $\varnothing$, while $B_\zs$ outputs $b$ with probability $p(b|y,\zs)$. 

The correlations obtained by this mixture of three strategies are
\begin{equation}\label{eq:srq-corr}
    \begin{aligned}
        p^{\text{srq}}(a,b|x,y,\zs) &= s \left(\frac{t}{m_A} + 1-t \right) p(a,b|x,y,\zs) ,\\
        p^{\text{srq}}(a,\varnothing|x,y,\zs) &= 0 ,\\
        p^{\text{srq}}(\varnothing,b|x,y,\zs) &= \left( s ~t ~ \frac{m_A - 1}{m_A}   + 1-s \right) p(b|y,\zs) ,\\
        p^{\text{srq}}(\varnothing,\varnothing|x,y,\zs) &= 0 ,\\
        p^{\text{srq}}(a,b|x,y,\zl) &= s \left(\frac{t}{m_A} + \frac{1-t}{m_{B_\zl}}  \right) p(a,b|x,y,\zl) ,\\
        p^{\text{srq}}(a,\varnothing|x,y,\zl) &= s (1-t) \left( \frac{m_{B_\zl} - 1}{m_{B_\zl}}  \right) p(a|x) ,\\
        p^{\text{srq}}(\varnothing,b|x,y,\zl) &= s~ t \left( \frac{m_A - 1}{m_A}  \right) p(b|y,\zl) ,\\
        p^{\text{srq}}(\varnothing,\varnothing|x,y,\zl) &= 1-s.
    \end{aligned}
\end{equation}
These correlations match the correlations $p^{\vec{\eta}}$ given in \eqref{eq:pobs_nobinning} when
$\eta_{B_\zs}=1$ and $\eta_{B_\zl}$ is equal to the right hand side of \eqref{eq:universal-bound-routed}.
It is of course possible in the above model to locally lower the detection efficiencies of $B_\zs$ and $B_\zl$ by instructing part of the time the measurement device to output $\varnothing$, thereby achieving arbitrary values of $\eta_{B_\zs}$ and any $\eta_{B_\zl}$ satisfying \eqref{eq:universal-bound-routed}.
\end{proof}

The above bound places fundamental limits on the distance at which nonlocal correlations can be observed, both for standard and routed Bell experiments.
In particular, the right-hand side of \eqref{eq:universal-bound-routed} is always greater than $1/m_{B_\zl}$, implying that the detection
efficiency of the remote measurement device cannot be lower than $1/m_{B_\zl}$, even when all other detectors are perfect. This lower-bound of $1/m_{B_\zl}$ can also be seen as a consequence of the universal bounds derived in \cite{Masini2023} for general (semi-)device-independent protocols.
In the case where Bob's remote measurement device is doing two measurements, as for the explicit examples considered in this paper and in \cite{Chaturvedi2022},  the detection efficiency of the remote measurement device cannot be lower than $1/2$. 

Though the bound \eqref{eq:universal-bound-routed} applies both to standard and routed Bell experiments, this does not mean that routed Bell experiment cannot be more robust to photon losses than standard Bell experiment. First, the above bound is universal and applies to any quantum strategy. For specific strategies, based on specific states and quantum measurements, routed Bell experiment may provide an advantage over standard Bell experiments, as we will demonstrate in the next subsection.

Second, a more stringent bound than \eqref{eq:universal-bound-routed} might actually apply for standard Bell experiments, leaving the possibility of a gap between the optimal (i.e., optimized over all quantum strategies) photon loss resistance of standard and routed Bell experiments.
In the case of two measurements per party, $m_A=m_{B_\zl}=2$, however, this cannot be the case. Indeed, the bound \eqref{eq:universal-bound-routed} then reduces to
\begin{equation}
    \eta_{B_\zl} \le \frac{\eta_A}{3\eta_A - 1} \, . \label{eq:universal-bound-unrouted-2}
\end{equation}
It was shown in \cite{Cabello2007} that there are quantum implementations that violate the CH inequality \cite{Clauser1974}, and thus demonstrate standard nonlocality, whenever \eqref{eq:universal-bound-unrouted-2} is violated, i.e., the bound \eqref{eq:universal-bound-unrouted-2} is tight for standard Bell experiments -- and routed Bell experiments cannot improve it.
Whether there exists a quantum strategy that demonstrates standard nonlocality when $\eta_B$ violates \eqref{eq:universal-bound-routed} for general values of $m_{A}$ and $m_{B_\zl}$ is an open question. 

\subsection{Analytical detection thresholds for implementations based on an ideal short-path CHSH test}
We now derive the detection efficiency threshold $\eta_{B_\zl}$ of the remote measurement devices for specific quantum correlations. We focus on the natural implementation depicted in Fig.~\ref{fig:Routed_Bell_efficiencies} in which the close detectors have the same efficiency $\eta_A = \eta_{B_\zs}=\eta_\zs$ and the distant detector has a lower efficiency $  \eta_{B_\zl} = \eta_\zl \leq \eta_\zs $. We assume, as in the previous section, that the state produced by the source and the observables implemented by the nearby devices $A$ and $B_\zs$ are
\begin{equation}\label{eq:Simpleqmstrat}
    \rho_{AB} = \ketbra{\phi_+}\,, \quad A_0 = X\,,\,A_1=Z\,,\quad  B_{0\zs} = \frac{X+Z}{\sqrt{2}}, \quad B_{1\zs} = \frac{X-Z}{\sqrt{2}}\,,
\end{equation}
yielding (in the ideal case $\eta_\zs=1$) a maximal CHSH violation in the short path. For the remote device $B_\zl$, we consider two families of strategies. Those, again as in the previous section, where $B_{0\zl}$ and $B_{1\zl}$ anti-commute, i.e., 
\begin{equation}
    B_{0\zl} = s_\theta X+c_\theta Z  ,~ B_{1\zl} = c_\theta X- s_\theta Z  ~~~(\text{anti-commuting $B_\zl$})\label{eq:anticommuting-BL}\,,
\end{equation}
where $c_{\theta} = \cos\theta$ and $s_{\theta} = \sin\theta$ and $\theta\in [0,\pi/4]$. These strategies include the CHSH strategy $(\theta = \pi/4)$ and the BB84 strategy $(\theta = 0)$ as special cases.

The second families of strategies are those where $B_{0\zl}$ and $B_{1\zl}$ correspond to arbitrary measurement in the $X-Z$ plane, which we parametrize as
\begin{align}
    B_{0\zl} &=   s_{\theta_0} X+c_{\theta_0} Z,~ B_{1\zl} =  s_{\theta_1} X + c_{\theta_1}Z ~~~(\text{general $B_\zl$})\,, \label{eq:general-BL}
\end{align}
where $\theta_0 = \theta_+-\theta_-$ and $\theta_1 = \theta_++\theta_-$. Thus $\theta_-$ corresponds to half the angle between the two measurement directions of $B_\zl$ in the Bloch sphere, and $\theta_+$ corresponds to a global rotation. Without loss of generality (by relabelling the outcomes of the measurements), we can restrict the angles $\theta_-$ to the interval $]0, \pi/2[ $ and $\theta_+$ to the interval $[0, \pi/4]$. The anti-commuting strategies \eqref{eq:anticommuting-BL} are a special case of the general strategies \eqref{eq:general-BL} with $\theta_- = \pi/4$ and $\theta_+ = \theta+\pi/4$.

\begin{figure}[t]
    \centering
    \includegraphics{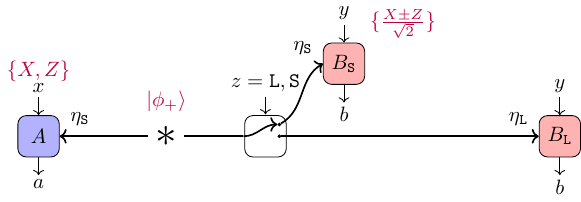}
    \caption{Given some target quantum strategy and some detector efficiency of the close parties $\eta_\zs \equiv \eta_{B_\zs} = \eta_{A}$ we look for the critical efficiency $\eta_\zl \equiv \eta_{B_\zl}$ required to certify nonlocality between the distant parties. \label{fig:Routed_Bell_efficiencies}}
\end{figure}

We will now derive analytically the critical efficiencies $\eta_\zl$ necessary to exhibit long-range quantum correlations when the nearby measurement devices are perfect, i.e., $\eta_\zs = 1$. In the next section, we will use numerical methods for the general case $\eta_\zs < 1$.

\subsubsection{Binned no-click outcomes}
Let us start by considering anticommuting measurements \eqref{eq:anticommuting-BL} for the remote device $B_\zl$ and that the no-click outcome $\varnothing$ is binned with the outcome $+1$. Then a necessary and sufficient condition for the long-path quantum correlations $p(a,b|x,y,\zl)$ to be nonlocal according to the standard notion of nonlocality (i.e, ignoring the switch) is to violate the CHSH inequality $\chsh_\zl > 2$. The quantum implementation given by eqs. \eqref{eq:Simpleqmstrat} and \eqref{eq:anticommuting-BL} yields for $\eta_\zs=1$ the CHSH value $\mathcal{C}_{\zl} = 2 \eta_\zl (c_\theta + s_\theta)$, which violates the local bound when
\begin{align}
    \eta_\zl > \frac{1}{c_\theta + s_\theta}\,. \label{eq:crit-eff-CHSH} 
\end{align}
The required efficiency thus goes from $\eta_\zl > 1$ when $\theta = 0$ (i.e., BB84 correlations do not exhibit nonlocality) to $\eta_\zl > 1/\sqrt{2} \approx 0.71 $ when $\theta = \pi/4$ (CHSH correlations).

Let us now exploit the extra constraints following from the short-path correlations. Since $\mathcal{C}_\zs=2\sqrt{2}$, SRQ correlations satisfy the \LP inequality $\mathcal{J}_{\zl}^{\theta} \le \sqrt{2}/c_\theta$ given in \eqref{eq:SimpleJM-ineq-bipartite}. The value of $\mathcal{J}_{\zl}^{\theta}$ 
for the anti-commuting implementations \eqref{eq:anticommuting-BL} is  $\mathcal{J}_{\zl}^{\theta} = 2 \eta_\zl / c_\theta \,$. 
Long-range quantum correlations are certified when this value exceeds the SRQ bound, i.e., when
\begin{equation}\label{eq:0.71}
\eta_\zl > {1}/{\sqrt{2}} \approx 0.71 \qquad \text{for all } \theta\,. 
\end{equation}
Thus, all the anti-commuting implementations \eqref{eq:anticommuting-BL} have the same tolerance to detection losses, which is moreover equal to the tolerance of the maximally loss-tolerant CHSH correlations.

We can reduce the critical efficiency in a routed Bell experiment further by using the following \LP inequality
\begin{multline}
    J^{{\theta_+},{\theta_-}}_{\zl} = (c_{\theta_+} + s_{\theta_-} s_{\theta_+}) \langle A_1B_{0\zl} \rangle 
    + (c_{\theta_+} -s_{\theta_-} s_{\theta_+}) \langle A_1 B_{1\zl} \rangle +
         (s_{\theta_+}-s_{\theta_-} c_{\theta_+}) \langle A_0 B_{0\zl} \rangle +\\
         (s_{\theta_+}+s_{\theta_-} c_{\theta_+}) \langle A_0 B_{1\zl} \rangle + 
        c_{\theta_-}( \langle B_{0\zl} \rangle +  \langle B_{1\zl} \rangle ) \le 2, \label{eq:JM-ineq-bipartite}
\end{multline}
where the SRQ bound $J^{{\theta_+},{\theta_-}}_{\zl} \le 2$ is derived assuming $\mathcal{C}_\zs = 2\sqrt{2}$ (see Appendix \ref{app:JM-proof} for a proof).
The value of $J^{{\theta+\pi/4},{\pi/4}}$ for the anti-commuting implementation  \eqref{eq:anticommuting-BL} is $\eta_\zl +\sqrt{2}$ which violates the SRQ bound when
\begin{align}
    \eta_\zl &> 2-\sqrt{2} \approx 0.586 \qquad\text{for all }\theta \label{eq:crit-eff-J-anti-commuting},
\end{align}
which represents a significant improvement over \eqref{eq:0.71}.

Let us now consider general projective measurements for $B_\zl$ of the form \eqref{eq:general-BL}. We then find that the CHSH value is $\mathcal{C}_{\zl} = 2 \eta_\zl {c_{\theta_+}(c_{\theta_-} + s_{\theta_-})}$, which violates the standard local bound when
\begin{align}
    \eta_\zl> \frac{1}{c_{\theta_+}(c_{\theta_-} + s_{\theta_-})}\,.\label{eq:crit-eff-CHSH-general}
\end{align}
This critical efficiency is plotted in Fig.~\ref{fig:crit-eff} as a function of $\theta_-$ for different values of $\theta_+$.
\begin{figure}[t]
    \centering
    \includegraphics[width=0.5\textwidth]{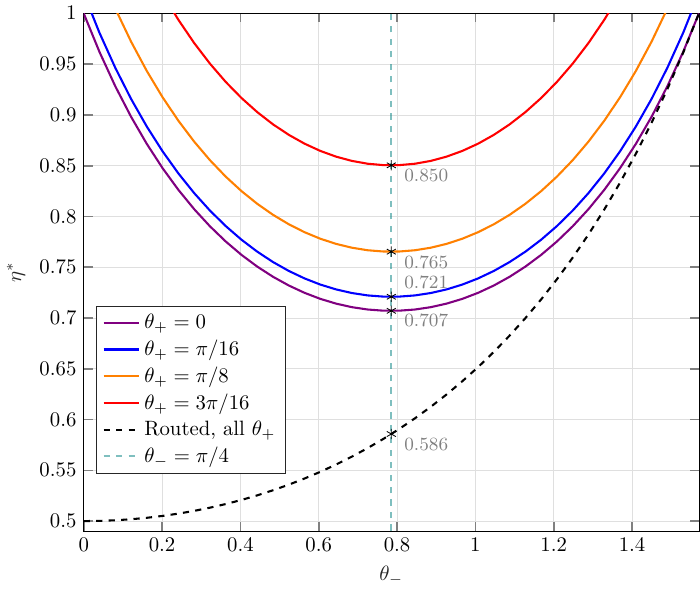}
    \caption{Critical efficiency $\eta_\zl$ for the quantum implementations \eqref{eq:general-BL}. The solid lines correspond to the critical efficiency in a standard Bell experiment, while the dashed line to the critical efficiency in a routed Bell experiment. Values for anti-commuting measurements, corresponding to $\theta_-=\pi/4$, are indicated on the graph.}
    \label{fig:crit-eff}
\end{figure}
If we instead use the SRQ inequality \eqref{eq:JM-ineq-bipartite}, we find that the value of $J^{{\theta_+},{\theta_-}}_{\zl}$ is $2  ( c_{\theta_-} + \eta_\zl  s_{\theta_-}^2 ) $, which violates \eqref{eq:JM-ineq-bipartite} when
\begin{align}
    \eta_\zl &> \frac{1}{1 + c_{\theta_-}}\,. \label{eq:crit-eff-J}
\end{align}
As $\theta_- \rightarrow 0$, this critical efficiency approaches $1/2$ (see Fig.~\ref{fig:crit-eff}),
which saturates the universal lower bound for the detection threshold \eqref{eq:universal-bound-routed}. We provide in Appendix \ref{app:explicit-strategy} an explicit SRQ strategy that 
reproduces the correlations obtained using \eqref{eq:general-BL} when $\eta = 1/(1 + c_{\theta_-})$, which proves that the bound \eqref{eq:crit-eff-J} is tight.

Note that as $\theta_-$ gets smaller, the target strategies become increasingly close to being jointly measurable, and the corresponding correlation close to SRQ correlations, since $B_{0\zl} \approx B_{1\zl}$. Somewhat counterintuitively, the implementations that are most robust to detection losses are precisely the ones that are the least robust against white noise, as happens for Eberhard correlations in the case of standard nonlocality \cite{Eberhard1993}.

\subsubsection{No-click outcomes kept as an additional output}
Let us now consider the situation where the no-click output $\varnothing$ of $B_\zl$ is kept as an additional outcome instead of binning it with $+1$. In a standard Bell scenario with two inputs and two outputs per party, this does not improve the analysis, as follows from \cite{Branciard2011}, i.e. the critical detection efficiencies for anticommuting and general measurements are still given by \eqref{eq:crit-eff-CHSH} and \eqref{eq:crit-eff-CHSH-general}, respectively.

In a routed Bell scenario, however, we do find an improvement by retaining non-detection events in the statistics. In fact, we find that for any anticommuting strategy from the family \eqref{eq:anticommuting-BL}, long-range quantum correlations can be certified whenever
\begin{equation}
    \eta_\zl > 1/2 \quad \text{for all } \theta. \label{eq:threshold-eta-TildeJ}
\end{equation}
This is optimal, as follows from the bound \eqref{eq:universal-bound-routed}. Thus, the maximally noise-robust anti-commuting measurements are also maximally loss-tolerant in a routed Bell experiment when non-detection events are included in the statistics.

To derive this result, we make use of the following \LP inequality
\begin{align}
    \tilde{\mathcal{J}}^{\theta}_{\zl} =  c_\theta \langle A_1 B_{0\zl} \rangle - s_\theta \langle A_1 B_{1\zl} \rangle + s_\theta \langle A_0B_{0\zl}\rangle + c_\theta \langle A_0B_{1\zl} \rangle - \frac{\langle T_{0\zl} \rangle + \langle T_{1\zl}\rangle}{2} \le \frac{1}{2}, \label{eq:TildeJM-ineq-bipartite}
\end{align}
where
  $T_{y \zl} = M_{b=+1|y\zl} + M_{b=-1|y\zl} $, and the SRQ bound $\tilde{\mathcal{J}}^{\theta}_{\zl} \le 1/2$ is derived assuming $\mathcal{C}_\zs = 2\sqrt{2}$ (see Appendix \ref{app:JM_tilde-proof} for proof).
The value of $\Tilde{\mathcal{J}}^{\theta}_\zl$ for the implementation \eqref{eq:anticommuting-BL} is $\Tilde{\mathcal{J}}^{\theta}_\zl= \eta$, which violates \eqref{eq:TildeJM-ineq-bipartite} when \eqref{eq:threshold-eta-TildeJ} holds.

The results obtained in this subsection are summarized in Table~\ref{table}.

\begin{table}[h]
    \centering
 
    \begin{tabular}{llll}
        \toprule
        Strategies & standard Bell (bin. or not bin.) & routed Bell, bin. & routed Bell, not bin. \\
        \midrule
        anti-commuting & $ 0.71\lesssim \frac{1}{c_\theta + s_\theta}\leq 1$ &  $2-\sqrt{2}\approx 0.586$ & $1/2$ \\
        general & $0.71\lesssim\frac{1}{c_{\theta_+}(c_{\theta_-} + s_{\theta_-})}\leq 1$ & $0.5\leq \frac{1}{1 + c_{\theta_-}}<1$ & $1/2^*$\\
        \bottomrule
    \end{tabular}
    \caption[Table - perfect close detectors]{Critical efficiency $\eta_\zl$ for the strategies given by eqs. \eqref{eq:Simpleqmstrat}, and \eqref{eq:anticommuting-BL} (anti-commuting) or \eqref{eq:general-BL} (general) in a standard Bell test or a routed Bell test in the case where the no-click outcome is binned or not. All values are derived analytically and proven to be tight, except for the value in the bottom-right corner which will be proved in an upcoming work \cite{future2024}.}
    \label{table}
\end{table}

\subsection{Numerical detection thresholds for implementations based on a lossy short-path CHSH test}
We now analyze numerically the setup considered in the previous section in the case where $\eta_\zs<1$, i.e., when the short-path CHSH violation is no longer maximal. We will focus on the anticommuting strategies given by \eqref{eq:anticommuting-BL} in the cases $\theta=0$, corresponding to BB84 correlations, and $\theta=\pi/4$, corresponding to CHSH correlations. Given a value $\eta_\zs$ for the short-path detector efficiency, we ask what is the maximal value $\eta_\zl$ for the long-path detector efficiency, so that the corresponding correlations $p^{\vec{\eta}}$ can be reproduced by an SRQ model, i.e., we aim to solve
\begin{align}
	\max_{\eta_{\zl}} \quad  \text{s.t. } p^{\vec{\eta}} &\in \mathcal{Q}_{\rm SR}\,.\label{eq:criteta}
\end{align}
Replacing the set $\mathcal{Q}_{\rm SR}$ by its $n$th-level NPA relaxation $\mathcal{Q}_{\rm SR}^n$, we can obtain an upper-bound on the critical long-path detection efficiency through
\begin{align}
	\max_{\eta_{\zl}} \quad  \text{s.t. } p^{\vec{\eta}} &\in \mathcal{Q}_{\rm SR}^n\,.\label{eq:critetaSDP}
\end{align}
Since the entries of the correlation vector $p^{\vec{\eta}}$ depend linearly on $\eta_\zl$ (once $\eta_\zs$ is fixed), the above problem is a SDP and can be solved numerically using standard packages. The results for a NPA level intermediate between the 3rd and the 4th are plotted in Fig.~\ref{fig:binning_vs_nobinning} (see Appendix \ref{appendix:sdp} for details). 


\begin{figure}[!htb]
    \centering
    \includegraphics{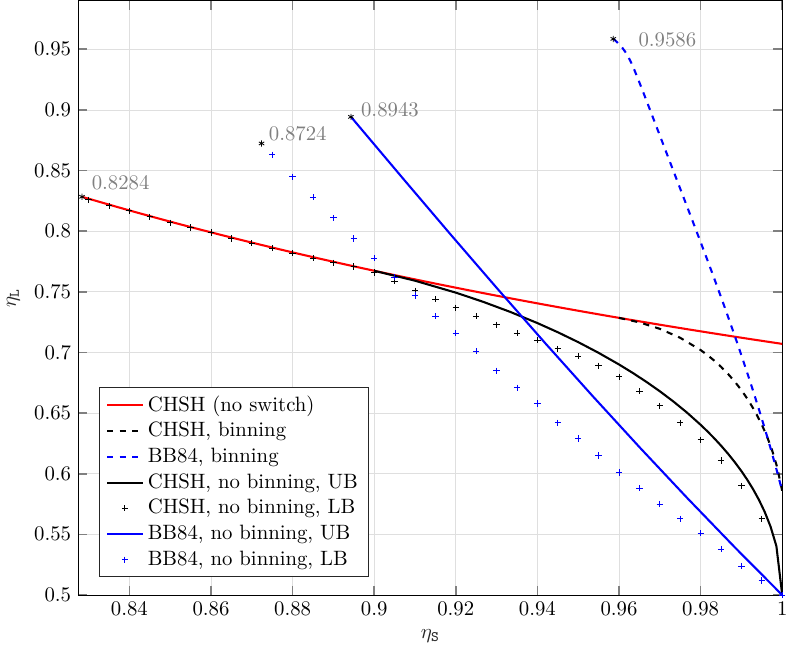}
    \caption{Upper and lower bounds on the critical efficiency $\eta_\zl$ for the distant device $B_\zl$ as a function of the efficiency $\eta_\zs$ for the devices $A$ and $B_\zs$ closer to the source, for the  CHSH correlations (black) and BB84 correlations (blue). The red curve correspond to the case of CHSH correlations in a standard Bell scenario, i.e., ignoring the switch (BB84 correlations in a standard Bell scenario are always local). Dotted (full) lines correspond to NPA upper bounds where the parties bin (keep) their no-click outcomes. The `$+$' values correspond to lower bounds where the parties keep their no-click outcomes. Lower-bounds in the case of binning are at most $1\%$ from the upper-bounds, hence not depicted. The upper bounds were obtained at level `$3 + AAAA+ B_\zs B_\zs B_\zs B_\zs + B_\zl B_\zl B_\zl B_\zl + A A B_\zs B_\zs + A A B_\zl B_\zl + B_\zs B_\zs B_\zl B_\zl $' of the NPA hierarchy, see Appendix~\ref{appendix:sdp} for details. We plot all curves until the point $\eta_\zl=\eta_\zs$. \label{fig:binning_vs_nobinning}}
\end{figure}

In the case of CHSH correlations, we find that routed Bell experiments provide a significant improvement over standard Bell test, especially for high-quality short-range tests. 
For a standard CHSH test, the critical efficiency is $\eta_\zl = \frac{\eta_\zs}{(\sqrt{2}+1)\eta_\zs - 1}$. This value is significantly reduced for routed Bell tests for short-path efficiencies above $\eta_\zs \sim 96.5 \%$ in the binning case and above $\eta_\zs \sim 91 \%$ in the no-binning case.

The BB84 correlations do not exhibit any nonlocality in standard Bell experiments, but, as already noted in the previous sections, do exhibit long-range quantum nonlocality in routed Bell experiments. Interestingly, for high values of $\eta_\zs$, BB84 correlations exhibit nonlocality for lower values of $\eta_\zl$ than CHSH correlations. The effect is much more pronounced in the no-binning case, but it is also present in the binning case (although it is too small to be visible in the figure).

We note that the critical efficiencies for $\eta_\zl$ obtained by solving \eqref{eq:critetaSDP} are not necessarily optimal, but only represent an upper-bound on the lowest admissible $\eta_\zl$, as we rely on a NPA relaxation at a finite level to compute them. We therefore also used heuristic seesaw algorithms to obtain lower-bounds on the critical efficiencies. Given efficiencies $\eta_\zs$ and $\eta_\zl$, we search for explicit SRQ strategies that reproduce the corresponding correlations $p^{\vec{\eta}}$. If we are able to do so, the given value $\eta_\zl$ represent a lower-bound on the critical efficiency necessary to exhibit long-range nonlocality. In the binning case, the lower-bounds obtained through our heuristic search over explicit quantum strategies are at most $1\%$ below the upper-bounds obtained through the NPA method and thus not represented in the figure. In the no-binning case, a small gap exists between the upper-bounds and lower-bounds, which are both plotted in Fig.~\ref{fig:binning_vs_nobinning}. Note that the existence of this gap does not affect our finding that BB84 correlations are more robust than CHSH correlations to losses in the long-path for high values of $\eta_\zs$.

Another important experimental consideration for Bell test implementations is noise. Here, we focus on the simplest case of local white noise: we assume that each party has local visibility $\nu_{B_z}=\nu_{A}\equiv\nu$. This can be modelled by replacing the state $|\phi_+\rangle$ in the ideal implementation by 
\begin{equation}
    \phi_{+}^\nu = \nu^2|\phi_{+}\rangle\langle\phi_+| + (1-\nu^2) \frac{\mathbf{1}}{4} \, ,
\end{equation}
We repeat the analysis for the no-binning case of the previous section for different local visibilities $\nu$. The results are shown in Fig.~\ref{fig:vis}. For CHSH we find numerically that the switch becomes useless for local visibilities below $\sim 0.940$. The BB84 correlations become local, even for perfect efficiency $\eta_\zs=1$ for local visibilities below $\sim 0.960$.

\begin{figure}[t]
    \centering
    \begin{subfigure}[b]{0.49\textwidth}
    \includegraphics[width=\textwidth]{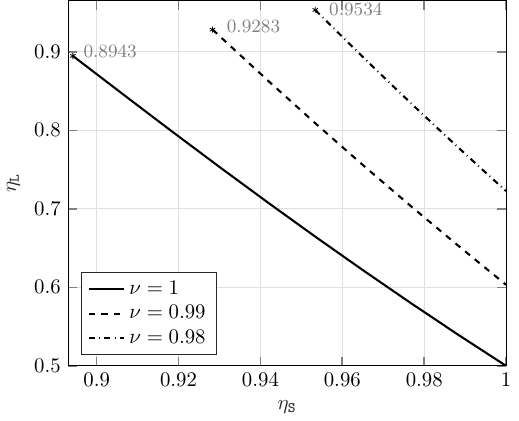}
    \caption{BB84 correlations without binning.}
    \end{subfigure}
    \hfill
    \begin{subfigure}[b]{0.49\textwidth}
    \includegraphics[width=\textwidth]{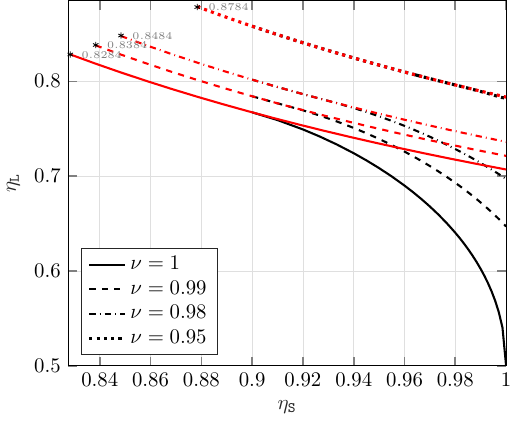}
    \caption{CHSH correlations without binning.}
    \end{subfigure}
    \caption{Upper bounds on the critical efficiency $\eta_\zl$ for the distant device $B_\zl$ as a function of the efficiency $\eta_\zs$ for the devices $A$ and $B_\zs$ closer to the source for different local visibilities $\nu$, when all parties keep their no-click outcomes. These bounds were obtained at the same NPA level as Figure~\ref{fig:binning_vs_nobinning}. For CHSH correlations the critical efficiencies for standard Bell tests (without the switch) are plotted in red. \label{fig:vis}}
\end{figure}

\section{Discussion}
We presented a systematic analysis of routed Bell experiments, an idea recently introduced by CVP \cite{Chaturvedi2022} with the goal of extending the distance over which nonlocality can be demonstrated in lossy experiments. We argued that certifying genuine nonlocality between the remote parties in routed experiments requires ruling out a more general class of `classical' models than those considered by CVP, which we termed short-range quantum models. Indeed, even if the behavior of the remote device might not be predetermined by classical variables set at the source, this does not imply that entanglement or quantum information needs to be distributed to it. 
After formulating the notions of short-range and genuine long-range correlations in routed experiments, we showed how these different correlation sets can be characterized through standard semidefinite programming hierarchies. With these definitions, we find that a short-path CHSH test does not lower the requirements for witnessing long-range nonlocality for a long-path CHSH test. However, borrowing intuition from self-testing and joint-measurability, we identified other long-path tests that do exhibit a `short-path enhancement', i.e. the requirements for observing a long-path quantum advantage are weakened. We then applied this analysis to the original question, namely whether these enhancements also imply improved detection efficiency requirements for nonlocality. With a combination of analytical and numerical tools, we showed that this is indeed the case, albeit the improvement is significantly lower than suggested by CVP's analysis. In particular, we showed that the efficiencies in routed Bell experiments are still subject to the same universal bounds valid in standard Bell experiments \cite{Massar2003,Masini2023}. When both parties have two measurement inputs, these bounds are tight and can already be saturated in standard Bell experiments. As such, the best criticial efficiency for routed experiments cannot be lower than the best (i.e., over all quantum strategies) critical efficiency for standard Bell experiments. Whether this is also the case for experiments with more inputs remains an open question.

Routed Bell experiments nevertheless do offer significant advantages. Indeed, for specific correlations, the introduction of a swich can be used to increase the robustness to losses and noise. For example, strategies exploiting anti-commuting Pauli measurements on a singlet state have single-sided critical efficiency that is at best $1/\sqrt{2} \sim 0.71$ in standard Bell experiments, but which is lowered to $1/2$ in routed experiments. This is the same single-sided critical efficiency as in the Eberhard scheme \cite{Eberhard1993}, but with a much simpler and more noise robust experimental setup. Furthermore, this bound can also be reached using the BB84 strategy, which involves both distant parties measuring $Z,X$ on the singlet. Interestingly, the BB84 correlations are always \emph{local} in the absence of the switch, i.e. the switch can be used to `activate' nonlocality. The idea that performing additional tests  in a Bell experiment can enlarge the set of correlations that can be certified in a device-independent way is similar in spirit to the results of \cite{bowles2018device} based on quantum networks, although routed Bell experiments do not require, contrarily to \cite{bowles2018device}, additional sources of entanglement or performing joint measurements.

The present work opens up several interesting directions for future research. First, it is natural to consider the symmetric extension of routed Bell experiments, where an additional \SP test is placed on both sides of the source. Our methods can be straightforwardly adapted to this situation. Clearly, this enables a further increase of the total distance over which nonlocality can be certified, but the extent of this increase will require further investigation \cite{Lobo2024}.
  
Second, it would be interesting to explore the practical applications of routed Bell experiments, in particular for cryptography and DIQKD. An improvement in the key rate and loss resistance with respect to conventional DI schemes is expected. Note, however, that the switch itself may be a significant source of losses and this has to be taken into account in any realistic analysis. Note further that the resulting performance of routed DI protocols cannot be better than analogous prepare-and-measure semi-DI schemes with a trusted preparation device. This is because at best the addition of the measurement device $B_\zs$ may lead to an exact certifcation of the entangled source and of Alice's device $A$, effectively turning them in a trusted remote preparation device. It follows that simple fully DI schemes based on CHSH or BB84 correlations will have key rates necessarily lower than or equal to CHSH or BB84 one-sided DI prepare-and-measure protocols \cite{woodhead2013quantum,tomamichel2013one}. In an upcoming work, we provide a detailed analysis of device-independent protocols based on routed Bell experiments \cite{Roy2024}.

Finally, it would be interesting to investigate how the idea of routed Bell experiments generalize to the more general setting of no-signalling, but supra-quantum, correlations. The analogues of CVP correlations in this context correspond to partially deterministic correlations introduced in \cite{ErikThesis,Bong2020,partdetpol}. It would be interesting to see if it is possible to define the analogues of short-range quantum correlations and whether `short-path enhancements' of long-path tests are also possible in the no-signalling setting.

\paragraph*{Noted added:} After presentation of an early version of this work at the \href{https://pirsa.org/23040123}{\emph{Causal Inference \& Quantum Foundations Workshop}} and the publication on the arXiv of the first version of the present paper, CVP updated their manuscript from \cite{Chaturvedi2022} to \cite{Chaturvedi2022v2} to take into account some of the points we raise.

\section*{Acknowledgements}
We thank Anubhav Chaturvedi, Giuseppe Viola, Marcin Pawłowski, Nicolas Gisin and Antonio Acín for useful discussions.
We acknowledge funding from the QuantERA II Programme that has received funding from the European Union's Horizon 2020 research and innovation programme under Grant Agreement No 101017733 and the F.R.S-FNRS Pint-Multi programme under Grant Agreement R.8014.21,
from the European Union's Horizon Europe research and innovation programme under the project ``Quantum Security Networks Partnership''(QSNP, grant agreement No 101114043),
from the F.R.S-FNRS through the PDR T.0171.22,
from the FWO and F.R.S.-FNRS under the Excellence of Science (EOS) programme project 40007526,
from the FWO through the BeQuNet SBO project S008323N,
from the Belgian Federal Science Policy through the contract RT/22/BE-QCI and the EU ``BE-QCI'' program.

S.P. is a Research Director of the Fonds de la Recherche Scientifique - FNRS. J.P. and E.P.L are FRIA grantees of the F.R.S.-FNRS.

Funded by the European Union. Views and opinions expressed are however those of the authors only and do not necessarily reflect those of the European Union. The European Union cannot be held responsible for them.

The code used in this work is available on GitHub at \url{https://github.com/eplobo/RoutedBell}.

\section{Appendix}
\appendix

\section{Local and quantum bounds on \texorpdfstring{$\mathcal{J}^{\theta}$}{}} \label{app:SimpleJM-bound-proof}
In the absense of a switch, $\mathcal{J}^{\theta}$ is a regular Bell expression and we can derive its local and quantum bounds using standard methods.
The bounds for $\theta \in [0,\pi/4]$ are given by \eqref{eq:SimpleJM-localbound} and \eqref{eq:SimpleJM-qmbound}, i.e.,
\begin{align}
    \mathcal{J}^{\theta} = t_\theta \langle A_0 B_{0} \rangle +  \langle A_0 B_{1} \rangle +  \langle A_1 B_{0} \rangle - t_\theta \langle A_1 B_{1} \rangle &\leq 2 \qquad\text{(local bound)} ~,  \label{appeq:SimpleJM-localbound} \\
            &\leq 2/c_\theta \qquad\text{(quantum bound)}\, . \label{appeq:SimpleJM-qmbound}
\end{align}
The local bound can be proved quite easily by explicitly checking the value of $\mathcal{J}^{\theta}$ for all possible local deterministic strategies, which corresponds to checking all possible assignments of $\pm 1$ to the observables $A_x$ and $B_y$. It is saturated by assigning $A_x = B_y = +1$.

To prove the quantum bound, on the other hand, takes a little more work. Let us denote the Bell operator corresponding to $\mathcal{J}^{\theta}$ by $\mathtt{J}^{\theta}$
\begin{align}
     \mathtt{J}^{\theta}  &=  t_\theta A_0  B_{0}  +   A_0  B_{1} +   A_1  B_{0} -  t_\theta A_1  B_{1}
\end{align}
 and assume without loss of generality (since we do not fix the dimension of the Hilbert space) that all the measurements are projective. Then, using that the observables $A_x, B_y$ are unitary operators which square to the identity, it is easily verified that
\begin{align}
    \left(c_{\theta} \mathtt{J}^{\theta}\right)^{2} &= 2\mathbf{1} +A_1A_0  ( c^2_{\theta} B_{0} B_{1} -  s^2_{\theta}  B_{1} B_{0}) +A_0A_1  ( c^2_{\theta} B_{1} B_{0} -  s^2_{\theta}  B_{0} B_{1}).
\end{align}
Denoting the the spectral norm by $\norm{\cdot}$, and using $\norm{(\cdot)^2} = \norm{\cdot}^2$, we obtain
\begin{align}
  c_{\theta}^2 \norm{\mathtt{J}^{\theta}}^2 = \norm{\left(c_{\theta}\mathtt{J}^{\theta}\right)^2}& \le 2 + \norm{A_1A_0  ( c^2_{\theta} B_{0} B_{1} -  s^2_{\theta}  B_{1} B_{0})} + \norm{A_0A_1  ( c^2_{\theta} B_{1} B_{0} -  s^2_{\theta}  B_{0} B_{1})} , \nonumber \\
   & \le 2 + \norm{A_1A_0} \cdot \norm{  c^2_{\theta} B_{0} B_{1} -  s^2_{\theta}  B_{1} B_{0}} + \norm{A_0A_1} \cdot \norm{   c^2_{\theta} B_{1} B_{0} -  s^2_{\theta}  B_{0} B_{1}},
\end{align}
where we have used that fact that for any two operators $P,Q$, $\norm{P + Q} \le \norm{P} + \norm{Q}$, and $\norm{PQ} \le \norm{P} \norm{Q}$. Further, since the observables are dichotomic, we have $\norm{A_0 A_1} \le \norm{A_0} \norm{A_1} \le 1 $, with similar bounds being valid for other pairs of observables. Therefore,
\begin{align}
    c_{\theta}^2 \norm{\mathtt{J}^{\theta}}^2& \le 2 + \norm{ c^2_{\theta} B_{0} B_{1}} + \norm{ s^2_{\theta}  B_{1} B_{0}} +  \norm{ c^2_{\theta} B_{1} B_{0}} + \norm{ s^2_{\theta}  B_{0} B_{1}}     \nonumber \\ 
   & \le 2 +   c^2_{\theta}  +  s^2_{\theta} +  c^2_{\theta}  +s^2_{\theta}   = 4  .
\end{align} 
Dividing by $c^2_{\theta}$ and taking the square root implies \eqref{appeq:SimpleJM-qmbound}. 

\section{SRQ bounds through joint-measurability inequalities} 
In this section we prove tight SRQ bounds on various \LP expressions of Sections \ref{sec:SPvsLP} and \ref{sec:dec}. These bounds are valid when the \SP CHSH test is maximal and are then equivalent to joint-measurability (JM) inequalities. 

\subsection{Bound on $\mathcal{J}^\theta_\zl$} \label{app:SimpleJM-proof}
The bound \eqref{eq:SimpleJM-srqbound} for $\mathcal{J}^\theta_\zl$ follows, as shown in Proposition~\ref{prop1},  from the following JM inequality
\begin{align}
    \frac{1}{2} \left[  \Tr(XB_{1\zl})+ \Tr(ZB_{0\zl})  \right] \le \sqrt{2}\,, \label{appeq:SimpleJM-ineq} 
\end{align}
valid for all qubit observables $B_{0\zl},B_{1\zl}$ that are jointly-measurable. This JM inequality was already introduced in \cite{Pusey2015}. We prove it below for completeness.
\begin{proof}
Restricting $B_\zl$ to be jointly-measurable means that $B_{0\zl}$ and $B_{1\zl}$ can be written as the marginals of a single parent POVM $\{N_{\beta_0 \beta_1}\}$, as shown in \eqref{eq:srqcorr4}. i.e.,
\begin{equation}
\begin{aligned}
	B_{0\zl} = M_{0|0} - M_{1|0} =   N_{00} + N_{01} - N_{10} - N_{11} \,, \\
	B_{1\zl} = M_{0|1} - M_{1|1} = N_{00} - N_{01} + N_{10} - N_{11} \,,
\end{aligned}
\end{equation}
where the outcome set $\{\pm 1\}$ has been relabeled $\{0,1\}$ for convenience. By rewriting \eqref{appeq:SimpleJM-ineq} in terms of $N_{\beta_0 \beta_1}$, we get
\begin{align}
    \frac{1}{2} \sum_{\beta_0 \beta_1} \Tr(C_{\beta_0 \beta_1} N_{\beta_0 \beta_1} ) \le \sqrt{2}, \label{appeq:SimpleJM-ineq-2}
\end{align}
where,
\begin{equation}
   C_{\beta_0 \beta_1} = (-1)^{\beta_0}Z + (-1)^{\beta_1} X.
\end{equation}
Proving \eqref{appeq:SimpleJM-ineq} is now equivalent to proving \eqref{appeq:SimpleJM-ineq-2}, with the restriction that the operators $N_{\beta_0 \beta_1}$ must be a valid POVM, i.e., $N_{\beta_0 \beta_1} \succeq 0$ and $\sum_{\beta_0 \beta_1}N_{\beta_0 \beta_1}= \mathbf{1}$. Maximizing the LHS of \eqref{appeq:SimpleJM-ineq-2} under these constraints is then an instance of semidefinite programming (SDP). The dual formulation of the SDP provides a certificate for the maximum value, which can be used to construct an analytic proof of \eqref{appeq:SimpleJM-ineq-2}. We provide such an analytic proof below which we constructed from the dual SDP.

Note first the following matrix inequalities, which can be easily verified by explicitly solving for the eigenvalues 
\begin{align}\label{eq:c}
    C_{\beta_0 \beta_1}  \preceq\sqrt{2} \mathbf{1} \qquad\forall \beta_0,\beta_1 .
\end{align}
We then get
\begin{align}
    \frac{1}{2} \sum_{\beta_0 \beta_1} \Tr(C_{\beta_0 \beta_1} N_{\beta_0 \beta_1} ) &\le  \frac{\sqrt{2}}{2} \sum_{\beta_0 \beta_1} \Tr(N_{\beta_0 \beta_1})\nonumber \,, \\
    & = \frac{\sqrt{2}}{2} \Tr(\mathbf{1}) = \sqrt{2}\,,
\end{align}
where in the first line we used \eqref{eq:c} and that $N_{\beta_0 \beta_1} \succeq 0$ and to get the second line we used that $\sum_{\beta_0 \beta_1}N_{\beta_0 \beta_1}= \mathbf{1}$.
\end{proof}
A joint-measurement strategy that saturates the inequality \eqref{appeq:SimpleJM-ineq} is given by $B_{0\zl}=B_{1\zl} = (Z+X)/\sqrt{2}$, which shows that the bound in \eqref{appeq:SimpleJM-ineq} is tight.
If $B_\zl$ is not restricted to be jointly-measurable, then the bound in \eqref{appeq:SimpleJM-ineq} can be improved to $2$ (which is the algebraic maximal value) using $B_{0\zl} = Z$, $B_{1\zl} = X$.

\subsection{Bounds on $\mathcal{J}_{\zl}^{\theta_+,\theta_-}$} \label{app:JM-proof}
We now derive the \LP inequality \eqref{eq:JM-ineq-bipartite} for ${J}_{\zl}^{\theta_+,\theta_-}$. 
Let us begin by considering the following joint-measurability inequality
\begin{align}
    \frac{1}{2} \left[ \Tr(Z B_{0\zl})  +
    \Tr(Z B_{1\zl}) - s_{\theta_-} \Tr(X B_{0\zl}) + s_{\theta_-} \Tr(X B_{1\zl}) + c_{\theta_-} \Tr(B_{0\zl}+B_{1\zl}) \right]  \le 2, \label{appeq:JM-ineq}
\end{align}
valid when $B_{0\zl}$ and $B_{1\zl}$ are jointly-measurable and $\theta_- \in ]0,{\pi}/{2}[$.
\begin{proof}
The proof is almost identical to Appendix \ref{app:SimpleJM-proof}. We write \eqref{appeq:JM-ineq} as
\begin{align}
    \sum_{\beta_0 \beta_1} \Tr(C_{\beta_0 \beta_1}({\theta_-}) N_{\beta_0 \beta_1} ) \le 2,
\end{align}
where,
\begin{equation}
    C_{\beta_0 \beta_1}({\theta_-}) = \delta_{\beta_0 \beta_1} (-1)^{\beta_0} (Z+c_{\theta_-} \mathbf{1}) + (1-\delta_{\beta_0 \beta_1}) (-1)^{\beta_1} s_{\theta_-} X.
\end{equation}
This expression plays a role analoguous to \eqref{appeq:SimpleJM-ineq-2} and as in Appendix~\ref{app:SimpleJM-proof} the bounds now follows from the following matrix inequalities which can easily be verified
\begin{align}
	C_{\beta_0 \beta_1}({\theta_-}) \preceq \mathbf{1} +  c_{{\theta_-}}   Z  \qquad \forall ~ {\beta_0, \beta_1}\,, \theta_- \in ]0,\frac{\pi}{2}[.
\end{align}
We can saturate the inequality \eqref{appeq:JM-ineq} by choosing $B_{0\zl} = B_{1\zl} = Z$.
\end{proof}
It is now easy to prove the \LP inequality \eqref{eq:JM-ineq-bipartite}. First note that we can view \eqref{appeq:JM-ineq} as part of a broader family of joint-measurability inequalities
\begin{multline}
    \frac{1}{2}[
        (c_{\theta_+} + s_{\theta_-} s_{\theta_+}) \Tr(Z B_{0\zl}) + 
        (c_{\theta_+} - s_{\theta_-} s_{\theta_+}) \Tr(Z B_{1\zl}) + 
        (s_{\theta_+} - s_{\theta_-} c_{\theta_+}) \Tr(X B_{0\zl}) + \\
         (s_{\theta_+} + s_{\theta_-} c_{\theta_+}) \Tr(X B_{1\zl}) + 
        c_{\theta_-} \Tr(B_{0\zl}+B_{1\zl}) 
        ] \le 2,
\end{multline}
that are  obtained by rotating the $X\shortminus Z$ plane by an angle $\theta_+$ around the $Y$ axis in the Bloch sphere. Such rotations correspond to a change of basis and do not change the bound. Using the equivalence shown in \eqref{eq:pmcorrelators} and valid when $\mathcal{C}_{\zs} = 2 \sqrt{2}$, we can interpret the family of joint-measurability inequalities above as the \LP inequalities \eqref{eq:JM-ineq-bipartite}.

Note that if the $B_{y\zl}$ observables are not restricted to be jointly-measurable, then the bound in \eqref{appeq:JM-ineq} can be improved to $2\sqrt{1+s_{{\theta_-}}^2}$, i.e.,
\begin{align}
    \frac{1}{2}\left[ \Tr(Z B_{0\zl})  +
    \Tr(Z B_{1\zl}) - s_{\theta_-} \Tr(X B_{0\zl}) + s_{\theta_-} \Tr(X B_{1\zl}) + c_{\theta_-} \Tr(B_{0\zl}+B_{1\zl})\right]  \le & 2\sqrt{1+s_{{\theta_-}}^2} . \label{appeq:JM-ineq-qm-bound}
\end{align}
\begin{proof}
To prove this, we follow a similar approach as before, but without the restriction to joint measurements. Let us re-express the above inequality as
\begin{align}
    \frac{1}{2} \sum_{b, y} \Tr(C_{b,y}({\theta_-}) M_{b|y\zl} ) \le ~2\sqrt{1+s_{{\theta_-}}^2},
\end{align}
where,
\begin{equation}
    C_{b,y}({\theta_-}) = (-1)^b (Z+c_{\theta_-} \mathbf{1}) + (-1)^{b + y +1} s_{\theta_-} X.
\end{equation}
Now, consider the following matrix inequalities, which can be easily verified
\begin{equation}
    \begin{aligned}
        C_{b,0}({\theta_-}) \preceq \lambda_0 \mathbf{1} + \lambda_1 X + \lambda_3 Z\,,\quad \forall b ,\\
        C_{b,1}({\theta_-}) \preceq \lambda_0 \mathbf{1} - \lambda_1 X + \lambda_3 Z\,,\quad\forall b,
    \end{aligned}
\end{equation}
where,
\begin{equation}
    \lambda_0 = \sqrt{1+s_{\theta_-}^2}, \quad
    \lambda_1 = -s_{\theta_-} + s_{\theta_-} \sqrt{ \frac{2-2c_{\theta_-} \sqrt{1+s_{\theta_-}^2}}{1+s_{\theta_-}^2}  } ,  \quad
    \lambda_3 = \frac{c_{\theta_-}}{\sqrt{1+s_{\theta_-}^2}} .
\end{equation}
These inequalities imply
\begin{align}
    \frac{1}{2} \sum_{b, y} \Tr(C_{b,y}({\theta_-}) M_{b|y\zl} ) &\leq  \frac{1}{2} \sum_{b, y} \Tr(  (\lambda_0 \mathbf{1} +(-1)^y \lambda_1 X + \lambda_3 Z)M_{b|y\zl} )\nonumber \,, \\
     &= \frac{1}{2} \sum_{y} \Tr(  (\lambda_0 \mathbf{1} +(-1)^y \lambda_1 X + \lambda_3 Z) ) = 2 \lambda_0 = 2 \sqrt{1+s_{{\theta_-}}^2} \,.
\end{align}
A strategy which saturates the inequality \eqref{appeq:JM-ineq-qm-bound} is given by, $B_{0\zl} + B_{1\zl} = \left(2/\sqrt{1+s_{{\theta_-}}^2}\right) Z$, \\ $ B_{1\zl} - B_{0\zl} = \left( 2s_{{\theta_-}} /\sqrt{1+s_{{\theta_-}}^2} \right) X$.
\end{proof}  
Finally, the local bound for $J^{\theta_+,\theta_-}$ in a standard Bell scenario (without the switch) is given by,
\begin{equation}
    J^{\theta_+,\theta_-} \le 2(c_{\theta_+} +s_{\theta_+}+c_{\theta_-}) ~\text{(local bound)}, \label{appeq:lhv-bound-J}
\end{equation}
where we have restricted $\theta_-$ to the interval $]0, \pi/2[$, and $\theta_+$ to the interval $[0, \pi/4]$. It can be saturated by $A$ and $B$ giving the output $+1$ irrespective of their inputs. When $s_{\theta_+}-s_{\theta_-} c_{\theta_+} \ge 0$, the value in \eqref{appeq:lhv-bound-J} is also the algebraic maximum of $J^{\theta_+,\theta_-}$, so it cannot be used in a standard Bell scenario to certify nonlocality.

\subsection{Bounds on $\tilde{\mathcal{J}}^{\theta}_{\zl}$} \label{app:JM_tilde-proof}
We now prove the \LP inequality \eqref{eq:TildeJM-ineq-bipartite} for $\tilde{\mathcal{J}}^{\theta}_{\zl}$, which is slightly more general than the ones above since there are three outcomes for  $B_\zl$. Let us begin by considering the following joint-measurability inequality 
\begin{align}
    \frac{1}{2}\Tr\left[Z B_{0\zl}  + X B_{1\zl} - \frac{(T_{0\zl} + T_{1\zl})}{2} \right] \le \frac{1}{2} \label{appeq:TildeJM-ineq},
\end{align}
where,
$~B_{y\zl} = M_{b=+1|y\zl} - M_{b=-1|y\zl}$,
$T_{y \zl} = M_{b=+1|y\zl} + M_{b=-1|y\zl}$\,.
\begin{proof}
    The proof is similar to Sections \ref{app:SimpleJM-proof} and \ref{app:JM-proof} except that there are three outcomes $\{+1,-1,\emptyset\}$ for $B_\zl$ in this scenario, which we relabel $\{0,1,2\}$ for convenience.
Hence, the joint-measurement $\{N_{\beta_0 \beta_1}\}$ of $B_\zl$ has, in general, nine outcomes. In terms of this joint-measurement, $B_\zl$'s observables are given by
\begin{equation}
    \begin{aligned}
        B_{0\zl} = M_{0|0}-M_{1|0} = N_{00} + N_{01} + N_{02} - N_{10} - N_{11} - N_{12}, \\
        B_{1\zl} = M_{0|1}-M_{1|1} = N_{00} + N_{10} + N_{20 }- N_{01}- N_{11} - N_{21}, \\
        T_{0\zl} = M_{0|0}+M_{1|0} = N_{00} + N_{01} + N_{02} + N_{10} + N_{11} + N_{12},\\
        T_{1\zl} = M_{0|1}+M_{1|1} = N_{00} + N_{10} + N_{20} + N_{01} + N_{11} + N_{21},
    \end{aligned} 
\end{equation}
which lets us rewrite \eqref{appeq:TildeJM-ineq} in terms of $N_{\beta_0 \beta_1}$ as
\begin{align}
		  \frac{1}{2} \sum_{\beta_0 \beta_1} \Tr(C_{\beta_0 \beta_1} N_{\beta_0 \beta_1} ) \le \frac{1}{2},
\end{align}
where,
\begin{equation}
    C_{\beta_0 \beta_1} = (1 -\delta_{\beta_0,2}) (-1)^{\beta_0} Z  + (1 -\delta_{\beta_1,2})(-1)^{\beta_1} X  - (2 -\delta_{\beta_0,2}-\delta_{\beta_1,2}) \frac{\mathbf{1}}{2}.
\end{equation}
It is now straightforward to verify the following matrix inequalities
\begin{equation}
	C_{\beta_0 \beta_1} \preceq \frac{\mathbf{1}}{2} \quad \forall ~\beta_0, \beta_1 \, ,
\end{equation}
using which \eqref{appeq:TildeJM-ineq} can be proved as before.

A joint-measurement strategy that saturates the inequality \eqref{appeq:TildeJM-ineq} is given by $M_{0|0}=\frac{1}{4}(\mathbf{1}+Z),~M_{1|0}=\frac{1}{4}(\mathbf{1}-Z),~M_{0|1}=\frac{1}{4}(\mathbf{1}+X),~M_{1|1}=\frac{1}{4}(\mathbf{1}-X)$, which can be shown to be 
jointly-measurable.
\end{proof}
It is now straightforward to prove the SRQ bound \eqref{eq:TildeJM-ineq-bipartite}. As before, we obtain a family of joint-measurability inequalities from \eqref{appeq:TildeJM-ineq} by rotating the $X-Z$ plane by an angle $\theta$ along the 
$Y$ axis in the Bloch sphere, given by
\begin{align}
    \frac{1}{2}\Tr\left[c_\theta Z B_{0\zl} - s_\theta Z B_{1\zl} + s_\theta X B_{0\zl} +
     c_\theta X B_{1\zl}  - \frac{(T_{0 \zl} + T_{1 \zl})}{2} \right] \le \frac{1}{2}, 
\end{align}
which, using the equivalence given in \eqref{eq:pmcorrelators} and valid when $\mathcal{C}_\zs=2\sqrt{2}$, gives the \LP inequality \eqref{eq:TildeJM-ineq-bipartite}.

 If $B_\zl$ is not restricted to be jointly-measurable, but can perform a general quantum measurement, then the bound in \eqref{appeq:TildeJM-ineq} can be improved to $1$:
\begin{align}
    \frac{1}{2}\Tr\left[Z B_{0\zl}  + X B_{1\zl} - \frac{(T_{0 \zl} + T_{1 \zl})}{2} \right] \le 1\,, \label{appeq:TildeJM-ineq-qm-bound}
\end{align}
\begin{proof}
    Let us write \eqref{appeq:TildeJM-ineq-qm-bound} as
\begin{align}
   \frac{1}{2} \sum_{b, y \in \{0,1\}} \Tr(C_{b,y} M_{b|y\zl} ) \le 1,
\end{align}
where,
\begin{equation}
    C_{b,y} = (-1)^{b} (\delta_{0,y} Z + \delta_{1,y} X ) - \frac{\mathbf{1}}{2}.
\end{equation}
Then the bound \eqref{appeq:TildeJM-ineq-qm-bound} follows from the following matrix inequalities
\begin{equation}
	C_{b,y} \preceq \frac{\mathbf{1}}{2} \quad\forall ~b,y \, .
\end{equation}
Choosing $B_{0\zl}=Z,~B_{1\zl}=X$ saturates the inequality \eqref{appeq:TildeJM-ineq-qm-bound}.
\end{proof}
The local and quantum bounds on $\tilde{\mathcal{J}}^{\theta}$ in a standard Bell scenario (without the switch) are given by
\begin{align}
    \tilde{\mathcal{J}}^{\theta} &\le 
    \begin{cases}
         2 c_\theta  - 1 & \text{ if } \theta \in \left[0,\frac{1}{2} \sin^{-1}(\frac{3}{4})\right]\\
          c_\theta + s_\theta - \frac{1}{2} & \text{ if } \theta \in \left[\frac{1}{2} \sin^{-1}(\frac{3}{4}),\frac{\pi}{4}\right]
    \end{cases} \qquad \text{(local bound)}, \label{appeq:TildeJM-ineq-local-bound} \\
    &\le 1 \qquad \text{(quantum bound)}.
\end{align} 
\begin{proof}
Let us first prove the local bound. Since $\tilde{\mathcal{J}}^{\theta}$ is linear in the expectation values, 
it must be maximized over a deterministic local strategy indexed by a local hidden variable, say,  $\lambda$.
Therefore,
\begin{align}
    \tilde{\mathcal{J}}^{\theta} \le \tilde{\mathcal{J}}^{\theta}_{\lambda} &=  \langle A_1\rangle_{\lambda} ( c_\theta  \langle B_{0} \rangle_{\lambda} -  s_\theta  \langle B_{1}\rangle_{\lambda} )  +  \langle A_0\rangle_{\lambda} (  s_\theta \langle B_{0}\rangle_{\lambda} +  c_\theta  \langle B_{1}\rangle_{\lambda} ) - \frac{\langle T_{0\zl}\rangle_{\lambda} + \langle T_{1\zl}\rangle_{\lambda}}{2},
\end{align}
where
\begin{equation}\label{appeq:tildeJM-corrs}
    \begin{aligned} 
        \langle A_x \rangle_{\lambda} &= \sum_{a \in \{ \pm 1\}} a~ p(a|x,\lambda) ~~,\forall x \in \{0,1\},  \\
        \langle B_y \rangle_{\lambda} &= \sum_{b \in \{ \pm 1\}} b~ p(b|y,\lambda) ~~,\forall y \in \{0,1\},  \\
        \langle T_y \rangle_{\lambda} &= \sum_{b \in \{ \pm 1\}} p(b|y,\lambda) ~~,\forall y \in \{0,1\}.
    \end{aligned}
\end{equation}
From \eqref{appeq:tildeJM-corrs} it is clear that $\absolutevalue{\langle B_y \rangle_{\lambda}} \le \langle T_y \rangle_{\lambda} $ and $\absolutevalue{\langle A_x \rangle_{\lambda}} \le 1$. Therefore,
\begin{align}
    \tilde{\mathcal{J}}^{\theta} &\le  \langle A_1\rangle_{\lambda} ( c_\theta  \langle B_{0} \rangle_{\lambda} -  s_\theta  \langle B_{1}\rangle_{\lambda} )  +  \langle A_0\rangle_{\lambda} (  s_\theta \langle B_{0}\rangle_{\lambda} +  c_\theta  \langle B_{1}\rangle_{\lambda} ) - \frac{ \absolutevalue{\langle B_{0\zl} \rangle_{\lambda}} + \absolutevalue{\langle B_{1\zl} \rangle_{\lambda}}}{2} ,\nonumber \\
    &\le \absolutevalue{( c_\theta  \langle B_{0} \rangle_{\lambda} -  s_\theta  \langle B_{1}\rangle_{\lambda} ) } + \absolutevalue{ (  s_\theta \langle B_{0}\rangle_{\lambda}+  c_\theta  \langle B_{1}\rangle_{\lambda} )} - \frac{ \absolutevalue{\langle B_{0} \rangle_{\lambda}} + \absolutevalue{\langle B_{1} \rangle_{\lambda}}}{2}.
\end{align}
The maximum of the right-hand side above occurs when $\langle B_{0} \rangle_{\lambda} = \langle B_{1} \rangle_{\lambda} = 1$ if
$\theta \in [0, (\sin^{-1}(3/4))/2]$, and when $\langle B_{0} \rangle_{\lambda} = 1,~ \langle B_{1} \rangle_{\lambda} = 0$ if
$\theta \in [(\sin^{-1}(3/4))/2,\pi/4]$. Substituting these values in the above inequality we get the desired bound,
which can be saturated using the follwing local strategy: When $\theta \in [0, (\sin^{-1}(3/4))/2]$, $A$ and $B$ always output $+1$. When $\theta \in [(\sin^{-1}(3/4))/2,\pi/4]$, $A$ always outputs $+1$ as before, but $B$ outputs $+1$ if $y=0$, and $\emptyset$ if $y=1$.

The quantum bound was determined numerically using level-1 of NPA hierarchy. 
It is tight, as it can be saturated by the quantum strategy given in \eqref{eq:Simpleqmstrat} and \eqref{eq:anticommuting-BL}, with $B_{\zl}$ playing the role of $B$.
\end{proof}

\section{Techniques for obtaining SoS decomposition}
\label{app:SOS-proof-methods}
In this section we describe the techniques we used to obtain the SoS decomposition for the family of shifted Bell operators $\mathtt{I}_u$
\eqref{eq:SOS-proof}. These techniques were first used in \cite{Bamps2015}, and we refer the interested reader to this work for more details. The steps are as follows.

We seek to write $\mathtt{I}_u = \sum_{i} K_i^\dagger K_i$ for some operators $K_i$. For this, expand $K_i$ in a basis of monomials $K_i = \sum_{ij} r_{ij} V_j$, where $V_j \in \{\mathbf{1}\} \bigcup \{A_0, ~ A_1\} \otimes \{B_{0\zs}, ~B_{1\zs},~ B_{0\zl},~ B_{1\zl}\}$. Therefore, $\mathtt{I}_u = \sum_{jk} V_j^\dagger  \left( \sum_{i} r_{ij}^* r_{ik} \right) V_k = 
\sum_{jk} V_j^\dagger  M_{jk}  V_k$.  Clearly, $M \succeq 0 $. Conversely, any positive operator $M$ that 
satisfies $\mathtt{I}_u = \sum_{jk} V_j^\dagger  M_{jk}  V_k$ is an SoS for $\mathtt{I}_u$. The aim is to find such a (non-unique) matrix $M$.

We can simplify the problem by utilizing known quantum strategies that saturate the bound on \eqref{eq:tradeoff-SRQvsLRQ-CHSHvsSimpleJ-2}. The existence of such strategies also establish that the bound \eqref{eq:tradeoff-SRQvsLRQ-CHSHvsSimpleJ-2} is tight.
For such strategies we must have, $\langle \mathtt{I}_u \rangle_{\psi} = 0$, and consequently, $\sum_{j} r_{ij} V_{j} \ket{\psi} = 0 ~\forall ~ i$.
Here $\ket{\psi}$ and $V_{j}$ are the known quantum state and measurement operators that yeild $\langle \mathtt{I}_u \rangle_{\psi} = 0$.

We made use of two such strategies which use the maximally entangled state $\ket{\phi_+}$ and the measurements
\begin{equation}
\begin{aligned}
\text{1. } A_0 &= c_u Z + s_u X,~ A_1 = c_u Z - s_u X,~ B_{0\zs} = Z,~ B_{1\zs} = X,~B_{0\zl}=B_{1\zl}=Z, \\
\text{2. } A_0 &= c_u X + s_u Z,~ A_1 = -c_u X + s_u Z,~ B_{0\zs} = Z,~ B_{1\zs} = X,~B_{0\zl}=-X,~B_{1\zl}=X.
\end{aligned}
\end{equation}
This imposes the following linear constraints on the coefficients $r_{ij}$:
\begin{equation}
\scalebox{0.98}{$
\begin{aligned}
r_{i \mathbf{1}} + c (r_{i A_0B_{0\zs}} + r_{i A_1B_{0\zs}} + r_{i A_0 B_{1\zl}} + r_{i A_0B_{0\zl}} + r_{i A_1B_{1\zl}} + r_{i A_1B_{0\zl}} ) + s (r_{i A_0B_{1\zs}} - r_{i A_1B_{1\zs}}) &= 0,  \\
s (-r_{i A_0B_{0\zs}} + r_{i A_1B_{0\zs}} - r_{i A_0 B_{1\zl}} - r_{i A_0B_{0\zl}} + r_{i A_1B_{1\zl}} + r_{i A_1B_{0\zl}}) + c (r_{i A_0B_{1\zs}} + r_{i A_1B_{1\zs}}) &= 0,  \\
r_{i \mathbf{1}} + c (r_{i A_0B_{1\zs}} - r_{i A_1B_{1\zs}} + r_{i A_0 B_{1\zl}} - r_{i A_0B_{0\zl}} - r_{i A_1B_{1\zl}} + r_{i A_1B_{0\zl}} ) + s (r_{i A_0B_{0\zs}} + r_{i A_1B_{0\zs}}) &= 0,  \\
s (r_{i A_0B_{1\zs}} + r_{i A_1B_{1\zs}} + r_{i A_0 B_{1\zl}} - r_{i A_0B_{0\zl}} + r_{i A_1B_{1\zl}} - r_{i A_1B_{0\zl}}) + c (r_{i A_1B_{0\zs}} - r_{i A_0B_{0\zs}}) &= 0.
\end{aligned}$} 
\end{equation}
Imposing these constraints, we find a new 5-dimensional monomial basis which must 
span $K_i$. The elements in this basis are $\{\mathtt{I}_u, P_1(u), P_2(u), P_3(u), P_4(u)\}$ (see~\eqref{eq:SOS-basis}).

A further simplification occurs by symmetry considerations. We note that under the 
transformation $\sigma_{1}=\{A_0\leftrightarrow A_1,~B_{1\zs} \rightarrow -B_{1\zs},~B_{0\zl} \leftrightarrow B_{1\zl}\}$, 
$\{\mathtt{I}_u,~P_1,~P_2\}$ are invariant, whereas $\{P_3,~P_4\}$ flip sign.
Applying this transformation to both sides of $\mathtt{I}_u = \sum_{jk} P_j^\dagger M_{jk}  P_k$, we get
\begin{equation}
\begin{aligned}
    \mathtt{I}_u &= \sum_{jk} \sigma_{1}(P_j^\dagger)  M_{jk}  \sigma_{1}(P_k), \\
        &= \sum_{jk} P_j^\dagger  N_{jk}  P_k ,
\end{aligned} 
\end{equation}  
where $N$ is a new SoS decomposition for $\mathtt{I}_u$ by construction.
The convex combination of SoS decompositions is also a valid SoS decomposition. 
A convenient choice is given by $G := (M+N)/2$ since it is block diagonal. In particular, 
$G$ is made up of $2$ blocks of sizes $3\times3$ and $2\times2$.

Similarly, we make use of another transformation $\sigma_{2} = \{A_1 \rightarrow -A_1, ~ B_{0\zs} \leftrightarrow B_{1\zs},~
B_{1\zl} \rightarrow -B_{1\zl} \}$, under which $\{\mathtt{I}_u,~P_1,~P_3\}$ are invariant and $\{P_2,~P_4\}$ flip sign. Imposing this symmetry to $G$, 
we find that we can choose yet another SoS decomposition $H$ that is block diagonal with $1$ block of size $2\times2$ and $3$ blocks of size $1\times1$.

Finally, we equate the coefficients on both sides of $\mathtt{I}_u = \sum_{jk} P_j^\dagger  H_{jk}  P_k$, and solve
for the non-zero elements of $H$, which leads to the SoS decomposition for $\mathtt{I}_u$ \eqref{eq:SOS-proof}.

\section{SRQ strategies that reproduce the correlations in  \texorpdfstring{\eqref{eq:general-BL}}{} with sufficiently small detection efficiency} \label{app:explicit-strategy}
Here we construct an explicity SRQ strategy for the quantum implementations in \eqref{eq:general-BL} that reproduces the statistics when the detection efficiency of $B_\zl$ is $1/(1 + c_{\theta_-})$, where ${\theta_-} \in ]0,\pi/2[$.
The task is to construct a joint measurement $N_{\beta_0 \beta_1}$, where $N_{\beta_0 \beta_1} \succeq 0 $ and $\sum_{\beta_0 \beta_1} N_{\beta_0 \beta_1} = \mathbf{1}$ for 
${\theta_-} \in ] 0,{\pi/2} [$, so that,
\begin{equation}
\begin{aligned}
    B_{0\zl}^\eta = \eta B_{0\zl} + (1-\eta) \mathbf{1} = N_{00} +  N_{01} -N_{10}-N_{11} , \\
    B_{1\zl}^\eta = \eta B_{1\zl} + (1-\eta) \mathbf{1} = N_{00} -  N_{01} +N_{10}-N_{11} ,
\end{aligned}
\end{equation}
where $\eta = 1/(1 + c_{\theta_-})$.
For $\theta_+ = 0$ the joint-measurement is given by 
\begin{equation}\label{appeq:joint-measurement}
    \begin{aligned}
        N_{00} &=  \left( \frac{c_{\theta_-}}{1+c_{\theta_-}} \right) (\mathbf{1}+Z)\,, \\
        N_{01} &= \frac{1}{2} \left[ -\left( \frac{s_{\theta_-}}{1+c_{\theta_-}} \right) X +  \left( \frac{1-c_{\theta_-}}{1+c_{\theta_-}} \right)  \frac{\mathbf{1}+Z}{2}
        +  \frac{\mathbf{1}-Z}{2} \right]\,, \\
        N_{10} &=   \frac{1}{2} \left[ \left( \frac{s_{\theta_-}}{1+c_{\theta_-}} \right) X +  \left( \frac{1-c_{\theta_-}}{1+c_{\theta_-}} \right)  \frac{\mathbf{1}+Z}{2}
        +  \frac{\mathbf{1}-Z}{2} \right] \,,\\
        N_{11} &= 0 \,  .
    \end{aligned}
\end{equation}
For other values of $\theta_+$, replace $Z$ and $X$ in \eqref{appeq:joint-measurement} with $c_{\theta_+} Z + s_{\theta_+} X$ and $c_{\theta_+} X - s_{\theta_+} Z$, respectively.

\section{SDP used in Fig.~\ref{fig:binning_vs_nobinning}}\label{appendix:sdp}

Here, we provide additional details about the implementation of the SDP problem \eqref{eq:critetaSDP}, which we solved to obtain the bounds on the critical detection efficiency $\eta_\zl$ in Fig.~\ref{fig:binning_vs_nobinning}. The semidefinite relaxations were generated using the python package NCPOL2SDPA \cite{Wittek2015, BrownGit} and solved using MOSEK \cite{mosek}. The code is available at  \url{https://github.com/eplobo/RoutedBell}.

We again focus on the CHSH and BB84 target correlations (\eqref{eq:anticommuting-BL} with $\theta=\pi/4$ and $\theta=0$ respectively). For the binning case, a basis of monomial operators is given by products of the hermitian observables $\{\mathbf{1},A_x,B_{yz}\}$ which satisfy $A_x^2=\mathbf{1}$, $B_{yz}^2=\mathbf{1}$, $[A_x,B_{yz}]=0$ and $[B_{0\zl},B_{1\zl}]=0$. The full set of probabilities $p^{\vec{\eta}}$ is in one-to-one correspondence with the expectation values  $\langle A_x\rangle$, $\langle B_{yz}\rangle$, and $\langle A_xB_{yz}\rangle$, which are given by 
\begin{equation}
\begin{aligned}
   \langle A_x\rangle &= 1 - \eta_\zs , \\
   \langle B_{yz}\rangle &= 1 - \eta_z ,\\
   \langle  A_xB_{y\zs}\rangle &= \eta_\zs^2 (-1)^{x \cdot y} \frac{1}{\sqrt{2}} + (1-\eta_\zs)^2 , \\
   \langle A_xB_{y\zl}\rangle &=  
   \begin{cases}
       \eta_\zs \eta_\zl (-1)^{x \cdot y} \frac{1}{\sqrt{2}} + (1-\eta_\zs)(1-\eta_\zl) \quad \text{for CHSH correlations}, \\
       \eta_\zs \eta_\zl \delta_{x,y} + (1-\eta_\zs)(1-\eta_\zl) \quad \text{for BB84 correlations} .
   \end{cases}  
\end{aligned}
\end{equation}
We solved the SDP problem \eqref{eq:critetaSDP} by imposing the above constraints on the expectations and considering level `$3 + AAAA+ B_\zs B_\zs B_\zs B_\zs + B_\zl B_\zl B_\zl B_\zl + A A B_\zs B_\zs + A A B_\zl B_\zl + B_\zs B_\zs B_\zl B_\zl $' of the hierarchy. 

When the no-click outcomes are not binned, the operators $A_x=M_{0|x}-M_{1|x}$ and $B_{yz}=M_{0|yz}-M_{1|yz}$ do not square to the identity, but we have the relations $A_x^2 = M_{0|x}+M_{1|x}$ and $A_x^3 = M_{0|x}-M_{1|x}=A_x$ and similarly for $B_{yz}$. We thus need to replace the polyonomial constraints $A_x^2 =\mathbf{1}$ and $B_{yz}^2=\mathbf{1}$ from the binning case by the constraints $A_x^3=A_x$ and $B_{yz}^3=B_{yz}$.  The full set of non-binned probabilities  $p^{\vec{\eta}}$ are in one-to-one correspondence with the expectation values  $\langle A_x\rangle$, $\langle A_x^2\rangle$, $\langle B_{yz}\rangle$, $\langle B_{yz}^2\rangle$, $\langle A_xB_{yz}\rangle$, $\langle A_xB_{yz}^2\rangle$, $\langle A_x^2 B_{yz}\rangle$, and $\langle A_x^2B_{yz}^2\rangle$, which, in our specific problem are given by
\begin{equation}\label{eq:expectations_nobinning}
\begin{aligned}
   \langle A_x\rangle &= 0, \quad \langle A_x^2 \rangle = \eta_\zs \, \\
   \langle B_{yz}\rangle &= 0, \quad \langle B_{yz}^2\rangle = \eta_z \,  \\
   \langle A_xB_{yz}^2\rangle &= \langle A_x^2 B_{yz} \rangle = 0 , \quad \langle A_x^2 B_{yz}^2 \rangle =\eta_\zs \eta_z \,\\
   \langle A_xB_{y\zs}\rangle &= \eta_\zs^2 (-1)^{x \cdot y} \frac{1}{\sqrt{2}} \, \\
   \langle A_x B_{y\zl} \rangle &= 
   \begin{cases}
       \eta_\zs \eta_\zl (-1)^{x \cdot y} \frac{1}{\sqrt{2}} \quad \text{for CHSH correlations} \,\\
       \eta_\zs \eta_\zl \delta_{x,y} \quad \text{for BB84 correlations}.
   \end{cases} 
\end{aligned}
\end{equation}
As before, we solved the corresponding SDP problem \eqref{eq:critetaSDP} at level `$3 + AAAA+ B_\zs B_\zs B_\zs B_\zs + B_\zl B_\zl B_\zl B_\zl + A A B_\zs B_\zs + A A B_\zl B_\zl + B_\zs B_\zs B_\zl B_\zl $'.

\printbibliography
\end{document}